\UseRawInputEncoding
\documentclass[12pt, draftclsnofoot, onecolumn]{IEEEtran}

\usepackage{graphicx,amssymb,amsmath}
\usepackage[noadjust]{cite}
\usepackage{setspace}               
\usepackage{stfloats}
\usepackage{amsthm}
\usepackage{amsmath}
\usepackage{flushend}
\usepackage{cases,subeqnarray}
\usepackage{bm,multirow,bigstrut}
\usepackage{textcomp}
\usepackage{latexsym,bm}
\usepackage{booktabs,changebar}
\usepackage{xcolor}
\usepackage{mathtools}
\usepackage{dsfont}
\usepackage{extarrows}
\usepackage{mathrsfs}
\usepackage{cite}
\usepackage{bm}
\usepackage{cleveref}
\usepackage{multicol}       
\usepackage{multirow}       
\usepackage{array}          
\usepackage{colortbl}
\usepackage{makecell}
\definecolor{crimson}{RGB}{192,0,0}         
\definecolor{navy}{RGB}{47,85,151}         

\makeatletter                               
\newif\if@restonecol
\makeatother

\makeatletter
\newif\if@restonecol
\makeatother

\usepackage[linesnumbered,ruled,vlined]{algorithm2e}
\usepackage{algpseudocode}
\usepackage{amsmath}
\theoremstyle{plain}

\newtheorem{lemm}{Lemma}

\newtheorem{coro}{Corollary}

\theoremstyle{plain}

\IEEEoverridecommandlockouts

\begin{document}

\title{A Survey on User-Centric Cell-Free Massive MIMO Systemss}

\author{Shuaifei~Chen, Jiayi~Zhang, Jing Zhang, Emil~Bj{\"o}rnson, and Bo Ai
\thanks{
This work was supported in part by the Swedish Research Council (VR) and Swedish Foundation for Strategic Research (SSF).}
\thanks{S. Chen, J. Zhang and J. Zhang are with the School of Electronic and Information Engineering, Beijing Jiaotong University, Beijing 100044, China. (e-mail: jiayizhang@bjtu.edu.cn).}
\thanks{E. Bj\"{o}rnson is with the Department of Electrical Engineering (ISY), Link\"{o}ping University, SE-58183 Link\"{o}ping, Sweden, and also with the Department of Computer Science, KTH Royal Institute of Technology, SE-16440 Kista, Sweden.
(e-mail: emilbjo@kth.se).}
\thanks{B. Ai is with State Key Laboratory of Rail Traffic Control and Safety, Beijing Jiaotong University, Beijing 100044, China. He is also with School of Information Engineering, Zhengzhou University, Zhengzhou 450001, China. (e-mail: boai@bjtu.edu.cn).}
}
\maketitle
\vspace{-2em}
\begin{abstract}
\vspace{-1.em}
The mobile data traffic has been exponentially growing during the last decades, which has been enabled by the densification of the network infrastructure in terms of increased cell density (i.e., ultra-dense network (UDN)) and/or increased number of active antennas per access point (AP) (i.e., massive multiple-input multiple-output (mMIMO)).
However, neither UDN nor mMIMO will meet the increasing data rate demands of the sixth generation (6G) wireless communications due to the inter-cell interference and large quality-of-service variations, respectively.
Cell-free (CF) mMIMO, which combines the best aspects of UDN and mMIMO, is viewed as a key solution to this issue.
In such systems, each user equipment (UE) is served by a preferred set of surrounding APs cooperatively.
In this paper, we provide a survey of the state-of-the-art literature on CF mMIMO.
As a starting point, the significance and the basic properties of CF mMIMO are highlighted.
We then present the canonical framework, where the essential details (i.e., transmission procedure and mathematical system model) are discussed.
Next, we provide a deep look at the resource allocation and signal processing problems related to CF mMIMO and survey the up-to-date schemes and algorithms.
After that, we discuss the practical issues when implementing CF mMIMO.
Potential future directions are then pointed out.
Finally, we conclude this paper with a summary of the key lessons learned in this field.
This paper aims to provide a starting point for anyone who wants to conduct research on CF mMIMO for future wireless networks.
\end{abstract}


\IEEEpeerreviewmaketitle

\section{Introduction}\label{sec:introduction}
{The performance of a mobile network is primarily quantified by the data rates that it can deliver to its users.
Since there is a multitude of user equipments (UEs) distributed over the coverage area, each experiencing unique propagation conditions, the rates that can be supported are highly user- and location-dependent. For example, the IMT-2020 requirements for the fifth generation (5G) technology specify downlink peak data rates that are 200 times larger than the so-called user-experienced data rates that should be guaranteed to 95\% of the users in the designated coverage area \cite{ITU-IMT2020}. When reducing the network performance into a single metric (e.g., to enable network dimensioning), it is common to consider the area traffic capacity, which is measured as the total data rate of all active users divided by the coverage area \cite{ITU-IMT2020}. When the network infrastructure is evolved to improve the area traffic capacity, the average data rates of the individual users will naturally increase, but it might have little impact on the user-experienced data rates, which are determined by the worst-case situations in the propagation environment. It is the user-experienced rates that determine which applications that can be utilized without interruption in the system, not the average or peak rates. Hence, to enable digitalization of society with a high perceived user fairness and consistent experience, the future network evolution should focus on the increasing the user-experienced data rates.}

As noticed by Cooper \cite{Cooper2010a}, the area traffic capacity of cellular networks is determined by the available bandwidth, physical-layer technology, and cell density.  In past decades, the vast majority of the improvements in traffic capacity is due to densification of the network infrastructure, in terms of increased cell density.
This is much in line with the original cellular philosophy \cite{Macdonald1979a}: the coverage area is divided into cells served by different access points (APs), so that the number of active users per cell is manageable for the AP.
Cellular networks were originally designed for voice services (i.e., mobile telephony), which are characterized by requiring a certain signal-to-noise ratio (SNR) to give an acceptable voice quality. If the SNR is below a threshold set by the codec, the voice is distorted and the call will eventually be dropped. As long as the SNR is above the threshold, the sound is distortion-free and, thus, the user experience is identical irrespective of how far above the threshold the SNR is.
Hence, the first generations of cellular networks could be dimensioned based on two principles: first, provide SNRs above the threshold almost everywhere in the coverage area; then, densify in regions where the number of active UEs is above what the APs can handle in the peak hours.

The situation is much different since mobile broadband became the dominant service in cellular networks, because the data rate increases continuously with the SNR \cite{cellfreebook}, up to the point where the maximum spectral efficiency (SE) is achieved.
Hence, cell densification has two positive impacts on the area traffic capacity of mobile broadband services: more UEs can be simultaneously active in the network, and their SNRs increase, which leads to a higher rate per UE.
Current networks consist of a mix of macro cells, micro cells, and small cells \cite{Hoydis2011c,Hwang2013a,Jungnickel2014a,Lopez2015c}.

5G features an additional type of densification: a large number of active antennas per AP, which is known as massive multiple-input multiple-output (mMIMO) \cite{Andrews2014a,Parkvall2017a}. This technology makes the radiation pattern of the APs highly adaptable and more directive, so that a larger fraction of the transmitted power reaches the region around the receiver, while there is less interference at undesired locations. Moreover, the technology allows for spatial multiplexing of UEs within each cell, if the UEs are located in sufficiently different parts of the cell \cite{Marzetta2010a,bjornson2017massive}. To enable efficient interference suppression in the spatial domain, a characteristic feature of mMIMO is that the AP has many more antennas than there are active UEs in the cell.
Broadly speaking, mMIMO has the same two positive impacts on the area traffic capacity as cell densification, but they are achieved differently. The benefit of mMIMO compared to cell densification is that {fewer APs are required to achieve a certain area capacity}, while the drawback is that each AP is equipped with more complicated hardware. So far, 5G makes use of small cells in millimeter-wave bands and mMIMO in sub-6 GHz bands.

The densification is expected to continue beyond 5G \cite{zhang2019multiple}, but both cell densification and mMIMO have fundamental limitations. As the cell area shrinks, the average SNR within a cell will improve, but the number of interfering cells also grows, which will eventually dominate. It is shown in \cite{Andrews2017a} that this effect is noticeable already when there are 10 APs per km$^2$.
This result relies on the assumption that all APs are transmitting simultaneously.
In the so-called ultra-dense network regime \cite{Chen2016a,Kamel2016a}, where there are many more APs than active UEs, only a random subset of the APs will have UEs to serve at any given point in time. The inactive APs will not cause interference in this regime, but since most APs will be idle most of the time, the required network infrastructure is utilized very inefficiently.
When it comes to mMIMO, this is a highly scalable technology in terms of the ability to multiplex many users spatially
 \cite{BjornsonHS17}; one can increase the array dimensions proportionally to the number of UEs that need to be served. However, the technology is rather inefficient in overcoming the large SNR differences that the UEs experience within macro and micro cells. In summary, both cell densification and mMIMO might be well suited for increasing the peak and average rates in future cellular networks. Still, the user-experienced data rates will remain modest due to inter-cell interference and large SNR variations.
\begin{figure}[t!]
\centering
\includegraphics[scale=0.8]{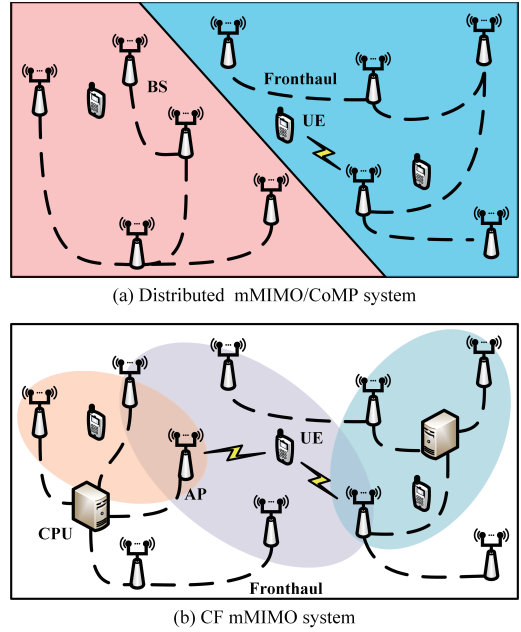}
\caption{Comparison of distributed mMIMO/CoMP and CF mMIMO.
\label{fig:system}}
\end{figure}
\begin{table*}[t!]
  \centering
  \fontsize{9}{9}\selectfont
  \caption{Comparison of cellular mMIMO, distributed mMIMO/CoMP, and CF mMIMO.}
  \label{tab:technology}
   \begin{tabular}{ !{\vrule width1.2pt}  m{3 cm}<{\centering} !{\vrule width1.2pt}  m{3 cm}<{\centering} !{\vrule width1.2pt}  m{4 cm}<{\centering} !{\vrule width1.2pt}  m{3.5 cm}<{\centering} !{\vrule width1.2pt}}

    \Xhline{1.2pt}
        \rowcolor{gray!50} \bf Technology  & \bf Cellular mMIMO  & \bf Distributed mMIMO/CoMP  & \bf CF mMIMO  \cr
    \Xhline{1.2pt}

        Coverage & Small  & Medium  & Large \cr\hline
        \multirow{3}{*}{ Clustering } & Network-centric  & Network-centric  & User-centric \\
        \cline{2-4} & Disjoint  & Disjoint  & Partially overlapping \\
        \cline{2-4} & Fixed  & Fixed  & Dynamic \cr\hline
        A UE served by  & One BS  & A few APs  & All surrounding APs \cr\hline
        CSI for decoding  & Instantaneous  & Instantaneous  & Instantaneous or statistical \cr\hline
        Fronthaul load  & --  & Large  & Small \cr\hline
        Synchronization & Uncritical  & Critical  & Critical  \cr\hline
        User-experience rate &  Low  & Medium  & Large  \cr

    \Xhline{1.2pt}
    \end{tabular}
  \vspace{0cm}
\end{table*}

\subsection{Cell-free mMIMO: The Best of Two Worlds?}

Cell-free mMIMO (CF mMIMO) is a new technology that basically combines the best aspects of ultra-dense cellular networks with the cellular mMIMO technology to overcome their respective weaknesses \cite{cellfreebook}. The name was coined in \cite{[2]} and refers to a network with many more APs than UEs and where the APs are cooperating to serve the UEs through coherent joint transmission and reception.
One way to picture it is to take a network containing a single mMIMO array, dismantle the array, and deploy the individual antennas at different locations while keeping the same transmission/reception algorithms.
When serving a given UE, the distributed antennas will then transmit each data signal with different power and phase-shifts, so they reach the intended UE synchronously and thereby reinforce each other. Similarly, the received signals at the different distributed antennas are co-processed to extract the data from each UE.
Another way to view the creation of the technology is to start from an ultra-dense cellular network, connect the APs to form a virtual distributed mMIMO array, as illustrated in Fig.~\ref{fig:system}(a), and then utilize (roughly) the same transmission algorithms as a conventional mMIMO array would do.

Irrespective of the direction from which one approaches the CF mMIMO technology, the main properties are that there are many geographically distributed APs, but the coverage area is not divided into disjoint cells.
Each UE is served by all the surrounding APs, as illustrated in Fig.~\ref{fig:system}(b). This mMIMO processing resolves the interference situation that limits conventional ultra-dense networks and leads to a network free from cells. Moreover, by having many distributed AP antennas instead of few APs with large antenna arrays, {the large SNR variations that limit the efficiency of conventional cellular mMIMO are effectively mitigated.}
The original motivation behind CF mMIMO was to design a new network infrastructure capable of providing uniform data rates in the coverage area \cite{[2]}; that is, concentrating on improving the user-experienced data rates, instead of the average or peak rates, which are already quite high in contemporary networks.

Since each UE will only be influenced by the signals from the closest surrounding APs, a CF mMIMO system can also be viewed as a user-centric network \cite{cellfreebook,[10],[11],[163]}. As illustrated in Fig.~\ref{fig:system}(b), each UE is served by a unique set of surrounding APs. To facilitate the cooperation between the neighboring APs flexibly, the technology has been conceived to make use of a cloud radio access network (C-RAN) infrastructure \cite{CRAN2011}. More precisely, the APs are connected via so-called fronthaul connections to one or multiple edge-cloud processors, which are called central processing units (CPUs) in the CF mMIMO literature \cite{[2]}.
{The backhaul connections can either be fully wired (e.g., using optical fiber cables) or partially wireless (e.g., using fixed microwave links).}

\begin{table}[t!]
  \centering
  \fontsize{9}{9}\selectfont
  \caption{Important Abbreviations.}
  \label{tab:symbol}
    \begin{tabular}{ !{\vrule width1.2pt}  m{1.6cm}<{\centering} !{\vrule width1.2pt}  m{5.7 cm}<{\centering} !{\vrule width1.2pt} m{1.6cm}<{\centering} !{\vrule width1.2pt}  m{5.7 cm}<{\centering} !{\vrule width1.2pt}}

    \Xhline{1.2pt}
        \rowcolor{gray!50} \bf Abbreviation & \bf Definition & \bf Abbreviation & \bf Definition  \cr
    \Xhline{1.2pt}
        5G    & fifth generation & IoT   & Internet of Things \\\hline
        6G    & Sixth generation & ISAC  & Integrated sensing and communication \\\hline
        ADC   & Analog-to-digital converter & LMMSE & Linear minimum mean-squared error \\\hline
        ANN   & Artificial neural network & L-MMSE & Local minimum mean-squared error \\\hline
        AoA   & Angle-of-arrival & LoS   & Line-of-sight \\\hline
        AoD   & Angles-of-departure & LP-MMSE & Local partial minimum mean-squared error \\\hline
        AP    & Access point & L-RZF & Local regularized zero-forcing \\\hline
        APO   & AP switch On/Off & LS    & Least-square \\\hline
        C\&F  & Compute-and-Forward & LSFD  & Large-scale fading decoding \\\hline
        CAP   & Compress-after-precoding & MEC   & Mobile edge computing \\\hline
        CBDNet & Convolutional blind denoising network & ML    & Machine learning \\\hline
        CDF   & Cumulative distribution function & mMIMO & Massive multiple-input multiple-output \\\hline
        CF    & Cell-free & MMSE  & Minimum mean-squared error \\\hline
        CFE   & Compress-forward-estimate & MR    & Maximum ratio \\\hline
        C-MMSE & Centralized minimum mean-squared error & MSE   & Mean-squared error \\\hline
        CoMP  & Coordinated multipoint & NLoS  & Non line-of-sight \\\hline
        CPU   & Central processing unit & NMSE  & Normalized mean-squared error \\\hline
        CP   & Cyclic prefix & OFDM  & Orthogonal frequency division multiplexing \\\hline
        C-RAN & Cloud radio access network & OTA   & Over-the-air \\\hline
        CS    & Compressive sensing & PAC   & Precoding-after-compress \\\hline
        CSI   & Channel state information & PA-MMSE & Phase-aware minimum mean-squared error \\\hline
        DCNN  & Deep convolutional neural network & P-FZF & Partial full-pilot zero-forcing \\\hline
        DFRC  & Dual-functional radar-communication & P-MMSE & Partial minimum mean-squared error \\\hline
        DFT   & Discrete Fourier transform & PWP-FZF & Protective weak partial full-pilot zero-forcing \\\hline
        E-C\&F & Expanded compute-and-forward & QoS   & Quality-of-service \\\hline
        ECF   & Estimate-compress-forward & RAU   & Radio access unit \\\hline
        EE    & Energy efficiency & RZF   & Regularized zero-forcing \\\hline
        EMCF  & Estimate-multiply-compress-forward & SCA   & Successive convex approximation \\\hline
        EMCFW & Estimate-multiply-compress-forward-weight & SE    & Spectral efficiency \\\hline
        EW-MMSE & Element-wise minimum mean-squared error & SIC   & Successive interference cancelation \\\hline
        FDD   & Frequency-division duplex & SINR  & Signal-to-interference-and-noise ratio \\\hline
        FFDNet & Flexible denoising convolutional neural network & SNR   & Signal-to-noise ratio \\\hline
        FL    & Federated learning & SOCP  & Second-order cone program \\\hline
        FZF   & Full-pilot zero-forcing & TDD   & Time-division duplex \\\hline
        GP    & Geometric programming & TOA   & Time-of-arrival \\\hline
        GPS   & Global position system & UatF  & Use-and-then-forget \\\hline
        HI    & Hardware impairments & UDN   & Ultra-dense network \\\hline
        ICA   & Independent component analysis & UE    & User equipment \\\hline
        IoE   & Internet of Everything & ZF    & Zero-forcing \\\hline

    \Xhline{1.2pt}
    \end{tabular}
  \vspace{0cm}
\end{table}

\subsection{Related Technologies}

The vision of serving UEs using multiple distributed APs has been around for a few decades. For example, Wyner described in \cite{Wyner1994a} from 1994 how one can untangle interfering uplink signals by joint detection at neighboring APs.
One can view this as a system of linear equations where each unknown variable is the information transmitted by one UE, and each equation is the signal received at one AP. A single-antenna AP can only identify one UE signal, but if the neighboring APs cooperate, they can jointly identify as many UE signals as APs.
A downlink counterpart of this concept was introduced by Shamai and Zaidel in \cite{Shamai2001a} from 2001.
Early embodiments of this technology have been called \emph{Distributed Wireless Communication System} \cite{Zhou2003a} and \emph{Network MIMO} \cite{Venkatesan2007a}.
Other prominent papers from this period are \cite{Jafar2004a,Zhang2004a,Zhang2005b,Tao2005a,Foschini2006a,Karakayali2006a,Tolli2008a,Sanderovich2009a,Simeone2009a,Khattak2008a,Bjornson2010c} and the survey article \cite{Gesbert2010a}. 3GPP called these technologies \emph{coordinated multipoint} (CoMP)  \cite{Parkvall2008a,Boldi2011a}.
The general premise was to evolve an existing cellular network by adding cooperation between the neighboring APs to reduce the inter-cell interference, not build a cell-free network from scratch, as is the vision with CF mMIMO.
Two main distinguishing factors between these early works and CF mMIMO are the operating regime with many more APs than UEs and the physical-layer operation inspired by the recent advancements in the mMIMO field.
For example, perfect channel state information (CSI) was generally assumed in the Network MIMO literature, and the methodology for analyzing the achievable data rates under imperfect CSI was largely missing at that time. When CoMP algorithms were analyzed under practical conditions, the gains were surprisingly low \cite{Boldi2011a}.
In practice, only the closest APs can acquire reliable CSI. Thus the system operation must be made robust to CSI imperfections, and there must be a resource-efficient way to acquire CSI. The mMIMO methodology provides these missing pieces.

The intended use case of the CoMP technology was to take an existing cellular network and divide the APs into disjoint clusters \cite{Choi2007a,Marsch2008a,Zhang2009b,Huang2009b}, which effectively creates a cellular network with distributed antennas within each cell cluster. {The significant difference between distributed mMIMO/CoMP and CF mMIMO is illustrated in Fig.~\ref{fig:system} and Table~\ref{tab:technology}.}
The CoMP technology reduces the SNR variations within each cluster but keeps the cellular structure. Thus UEs at the edges of a cell cluster are affected by interference from neighboring cell clusters. This is not the case in a CF mMIMO network, where every UE is served by all the surrounding APs.

\subsection{Contributions of This Survey}

CF mMIMO evolves from multiple existing techniques, like ultra-dense network (UDN), CoMP, Network MIMO, and mMIMO.
The convergence of these different technologies bursts out the novel strategies and procedures to meet the stringent requirements of the sixth generation (6G) network in terms of the high user-experienced data rates and ubiquitous coverage, which necessitates surveying the existing works on this topic and address the future research directions.

In this survey, we first describe in detail the technical foundations of CF mMIMO by giving a brief tutorial in a canonical framework.
This serves for a better understanding of the extensive survey on the state-of-the-art schemes and algorithms for resource allocation, signal processing, practical implementation, and future research directions on this topic.
To the best of the authors' knowledge, there has been no comprehensive survey on CF mMIMO available in the literature, although there are three survey and tutorial articles \cite{cellfreebook,[135],[103]} available, that however missed some parts of the holistic overview.
Specifically, \cite{cellfreebook} and \cite{[135]} provided tutorials on CF mMIMO while they lack a comprehensive survey on the applied schemes and algorithms in the literature, \cite{[135]} gave an early survey while extensive research has been conducted after that.
This motivates this survey to review the up-to-date works on CF mMIMO to provide a starting point for anyone who wants to conduct research on this topic.

\subsection{Paper Outline}

The remainder of this paper is organized as follows.
Section \ref{sec:foundation} introduces the technical foundations for CF mMIMO, where the transmission procedure and mathematical system model are discussed.
Section \ref{sec:processing} gives a comprehensive survey on resource allocation and signal processing while the practical issues when implementing CF mMIMO are discussed in Section \ref{sec:practical}.
Then, potential future directions of CF mMIMO research are highlighted in Section \ref{sec:future}.
{Finally, this paper is concluded Section \ref{sec:conclusion}} with a summary of the key lessons learned in this field.
The overall roadmap of this paper is illustrated in Fig.~\ref{fig:roadmap}.

\begin{figure*}[t!]
\centering
\includegraphics[scale=0.8]{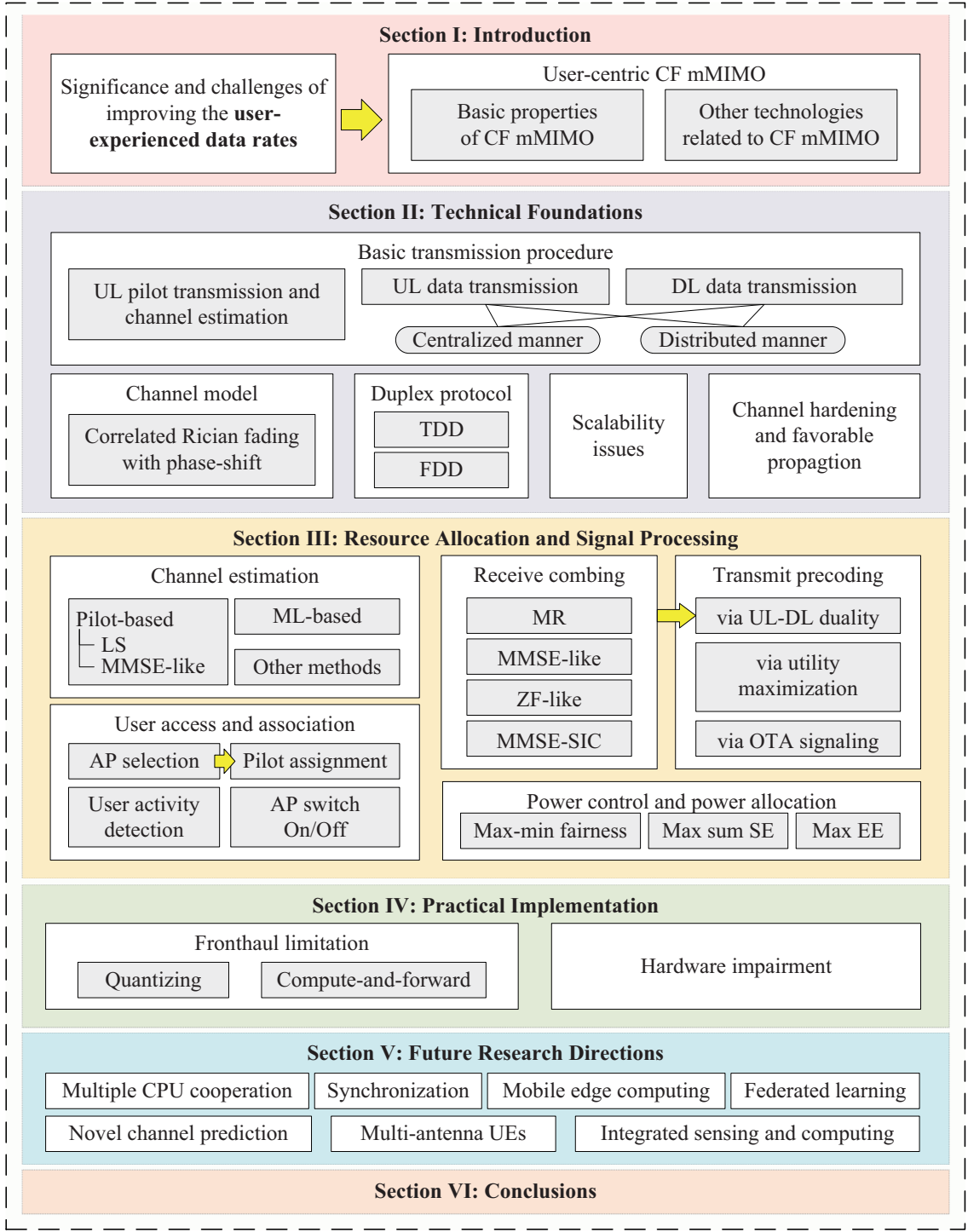}
\caption{Survey roadmap.
\label{fig:roadmap}}
\end{figure*}
\clearpage
\subsection{Notation}

Boldface lowercase letters, $\bf x$, denote column vectors and boldface uppercase letters, $\bf X$, denote matrices.
The superscripts $^{\rm T}$, $^{\rm *}$, and $^{\rm H}$ denote transpose, conjugate, and conjugate transpose, respectively.
The $n \times n$ identity matrix is ${\bf I }_n$.
We use $ \buildrel \Delta \over = $ for definitions and ${\rm {diag}}\left({{\bf A}_{1} ,\ldots, {\bf A}_{n}}\right)$ for a block-diagonal matrix with the square matrices ${{\bf A}_{1} ,\ldots, {\bf A}_{n}}$ on the diagonal.
The multi-variate circularly symmetric complex Gaussian distribution with correlation matrix $\bf R$ is denoted ${\cal N}_{\mathbb C}\left({{\bf 0},{\bf R}}\right)$.
{The expected value of $\bf x$ is denoted as ${\mathbb E}\left\{{\bf x}\right\}$.}
{We use $\left|{\cal A}\right|$ to denote the cardinality of the set $\cal A$.}

\begin{table}[tp]
  \centering
  \fontsize{9}{9}\selectfont
  \caption{Important Mathematical Symbols.}
  \label{tab:}
    \begin{tabular}{ !{\vrule width1.2pt}  m{1. cm}<{\centering} !{\vrule width1.2pt}  m{5.6 cm}<{\centering} !{\vrule width1.2pt}}

    \Xhline{1.2pt}
        \rowcolor{gray!50} \bf Notation  & \bf Definition     \cr
    \Xhline{1.2pt}

        $K$, $L$, $N$ & \makecell[c]{Number of the UEs, APs,\vspace{-0.2em} \\and antennas per AP}   \cr\hline
        ${\tau}_{p}$ & Number of orthogonal pilot sequences  \cr\hline
        $k$, $i$ & Index of the UEs \cr\hline
        $l$, $j$ & Index of the APs \cr\hline
        $t_k$ & Index of the pilot assigned to UE $k$ \cr\hline
        ${\cal M}_k$ & Subset of APs serving UE $k$ \cr\hline
        ${\cal D}_l$ & Subset of UEs served by AP $l$ \cr\hline
        ${\cal S}_k$ & \makecell[c]{Subset of UEs sharing pilot $t_k$,\vspace{-0.2em} \\including UE $k$} \cr\hline
        ${\cal P}_k$ & Subset of UEs served by partially the same APs as UE $k$, including UE $k$ \cr\hline
        ${\bf h}_{kl}$ & Channel response between UE $k$ and AP $l$ \cr\hline
        ${\hat {\bf h}}_{kl}$ & Channel estimate of ${\bf h}_{kl}$ \cr\hline
        $\beta_{kl}$ &  \makecell[c]{Large-scale fading coefficient of UE $k$\vspace{-0.2em} \\and AP $l$}  \cr\hline
        ${\bf v}_{kl}$ &  Combining vector that AP $l$ selects for UE $k$  \cr\hline
        ${\bf w}_{kl}$ &  Precoding vector that AP $l$ selects for UE $k$  \cr

    \Xhline{1.2pt}
    \end{tabular}
  \vspace{0cm}
\end{table}

\section{Technical Foundations}\label{sec:foundation}

Having explained the basic motivation and properties of CF mMIMO in the previous section, we will now take a look at its fundamental technical components and details. We first present the basic transmission procedure of CF mMIMO systems.
We then discuss the mathematical system model, including the models of the fading channels, duplex protocols, scalability issues, and channel hardening and favorable propagation.

\begin{figure*}[t!]
\centering
\includegraphics[scale=0.75]{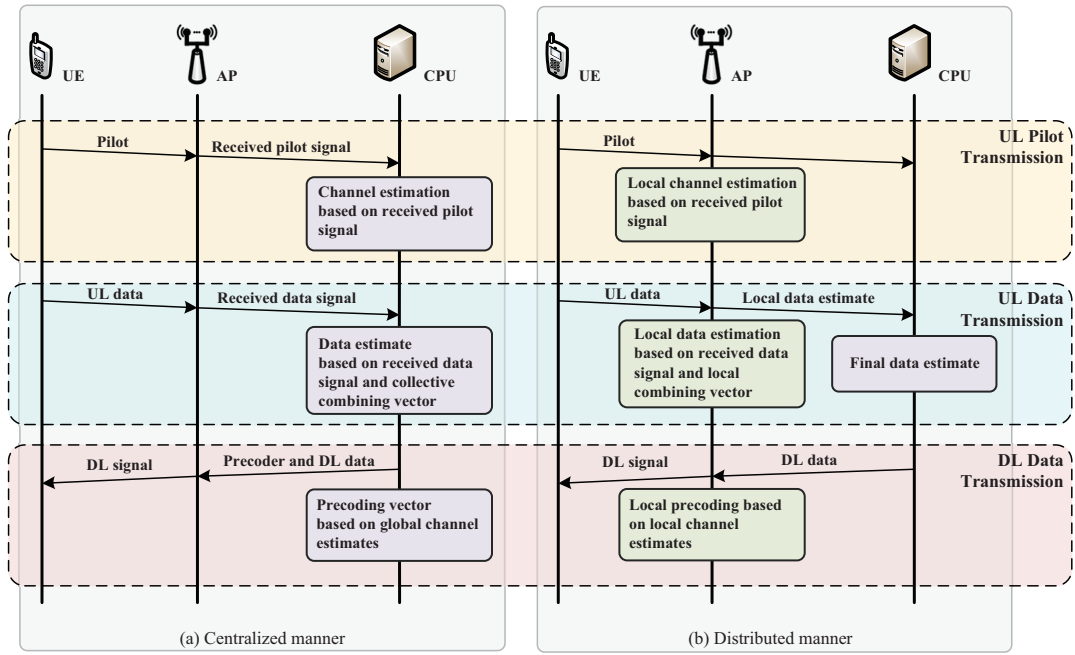}
\caption{Flow chart of three-stage transmission procedure for CF mMIMO.
\label{fig:procedure}}
\end{figure*}

\subsection{Basic Transmission Procedure}\label{subsec:procedure}

We consider a representative CF mMIMO system consisting of $K$ single-antenna UEs and $L$ APs, each equipped with $N$ antennas.
As illustrated in Fig.~\ref{fig:system}(b), all APs are connected to a CPU in an arbitrary fashion with the fronthaul connections. These connections facilitate the cooperation between the APs, such as the coherent joint transmission of the data signals to the UEs and the coherent joint reception of the data signals from the UEs.
The system could operate either in time-division duplex (TDD) mode or in frequency-division duplex (FDD) mode, which is further discussed in Section \ref{subsec:duplex protocol}.
For now, all APs and UEs are assumed to be operated in TDD mode. The propagation channels vary over time and frequency, which we describe using a block fading model \cite{cellfreebook}. More precisely, the time-frequency grid is divided into coherence blocks of $\tau_c$ channel uses for which the channel is constant and frequency flat. Each coherence block is divided into three phases: $\tau_p$ channel uses for uplink channel estimation, $\tau_u$ and $\tau_d$ for uplink and downlink data, respectively, such that $\tau_c = \tau_p + \tau_u + \tau_d$.
Note that this block-fading model is an abstraction of the practical multi-carrier modulation model used in orthogonal frequency division multiplexing (OFDM), and we consider it to make sure that the main concepts of CF mMIMO are not hidden under the more complicated notation that OFDM requires.
However, one can still map a practical OFDM system to this block-fading system by regarding several subcarriers, which compose the available bandwidth, as an aforementioned coherence block; see \cite[Sec.~2]{Marzetta2016a} for a concrete example.
We denote by ${{\bf{h}}_{kl}} \in {{\mathbb{C}}^N}$ the channel response between AP $l$ and UE $k$, which is a random realization in each coherence block of some stationary ergodic fading distribution.
Different fading distributions will be discussed in Section \ref{subsec:channel model}.

\subsubsection{Uplink Pilot Transmission and Channel Estimation}
When the UEs have gained access to the network, they are assigned pilots which are used for channel estimation.
We assume there are $\tau_p$ mutually orthogonal $\tau_p$-length pilot signals, where $\tau_p$ is a constant independent of $K$.
The pilot resources are limited due to the natural channel variations in the time and frequency domain. Thus we have $\tau_p < K$ in most practical scenarios, and the pilots must thus be reused between UEs.
Different algorithms for the pilot assignment are surveyed in Section \ref{subsec:access}.
For now, we denote by $t_k \in \{ 1,\ldots,\tau_p \}$ the index of the pilot assigned to UE $k$ and call ${\cal S}_k = \{ i : t_k = t_i \} \subset \{1,\ldots,K\}$ the subset of UEs sharing pilot $t_k$, including UE $k$.
When the UEs in ${\cal S}_k$ transmit pilot $t_k$, the received signal
${\bf y}_{{t_k}l}^{\rm pilot} \in {\mathbb C}^N$ after despreading at AP $l$ is \cite[Sec. 3]{cellfreebook}
\begin{equation}\label{eq:pilot}
  {\bf{y}}_{{t_k}l}^{{\rm{pilot}}} = \sum\limits_{i \in {{\cal S}_k}} {\sqrt {{\tau _p}{p_i}} {{\bf{h}}_{il}}}  + {{\bf{n}}_{{t_k}l}},
\end{equation}
where $p_i$ is the transmit power of UE $i$ and ${{\bf{n}}_{{t_k}l}} \sim {\cal N}_{\mathbb C}\left( {\bf 0},\sigma_{\rm ul}^2 {\bf I}_N \right)$ is the thermal noise.

By exploiting different levels of prior information regarding the fading distributions, different kinds of channel estimators could be utilized at the APs (or CPU) to estimate the channels based on \eqref{eq:pilot}, which are elaborated in Section \ref{subsec:estimation}.
For now, we denote by ${{{\bf{\hat h}}}_{kl}} \in {{\mathbb{C}}^N}$ the estimate of the channel ${\bf h}_{kl}$.
As shown in Fig.~\ref{fig:procedure}, an AP can perform the channel estimation locally or delegate this task to the CPU by just handing over its received pilot signals.
Since the pilots are reused between UEs, there will be interference between the pilot-sharing UEs. This will reduce the channel estimation quality, which makes the coherent transmission less effective. It also makes it harder to reject interference between the UEs sharing the same pilot since the AP cannot correctly separate their channels.
This is the so-called \emph{pilot contamination} phenomenon, which is particularly famous in the cellular mMIMO literature \cite{Marzetta2010a,Sanguinetti2019a}, but exists in any wireless system with non-orthogonal pilot transmissions.

\subsubsection{Uplink Data Transmission}

During uplink data transmission, the received signal ${\bf y}_{l}^{\rm ul} \in {\mathbb C}^N$ at AP $l$ is
\begin{equation}\label{eq:uplink received signal}
  {\bf{y}}_l^{{\rm{ul}}} = \sum\limits_{i = 1}^K {{{\bf{h}}_{il}}{s_i}}  + {{\bf{n}}_l},
\end{equation}
where $s_i \in {\mathbb C}$ is the signal transmitted from UE $i$ which is assumed to be transmitted with power $p_i$ and ${{\bf{n}}_{l}} \sim {\cal N}_{\mathbb C}\left( {\bf 0},\sigma_{\rm ul}^2 {\bf I}_N \right)$ is the noise.
One can estimate $s_k$ by properly combining the received signal ${\bf y}_{l}^{\rm ul}$.
For now, we let ${\hat s}_k$ denote the estimate of ${s}_k$ and call ${\bf v}_{kl} \in {\mathbb C}^N$ the combining vector that AP $l$ assigns to UE $k$.
Due to the three-stage architecture of CF mMIMO, there can be two levels of cooperation among APs for combing design, i.e., centralized combining and distributed combining as shown in Fig.~\ref{fig:procedure}.

{\bf Centralized Combining} In the first level, all APs forward their received data signals $\{ {\bf{y}}_l^{{\rm{ul}}} :l=1,\ldots,L \}$ to the CPU, which performs channel estimation and data detection in a fully centralized fashion.
Since no local signal processing is performed at APs, the CPU sees a collection of the received data signals
\begin{equation}\label{eq:collective uplink received signal}
  {{\bf{y}}^{{\rm{ul}}}} = \sum\limits_{i = 1}^K {{{\bf{h}}_i}{s_i}}  + {\bf{n}},
\end{equation}
where ${\bf y}^{\rm {ul}} = [ ({\bf y}_{1}^{\rm {ul}})^{\rm T},\ldots,({\bf y}_{L}^{\rm {ul}})^{\rm T} ]^{\rm T} \in {\mathbb C}^{LN}$, $ {\bf h}_i = [ {\bf h}_{i1}^{\rm T},\ldots,{\bf h}_{iL}^{\rm T} ]^{\rm T} \in {\mathbb C}^{LN}$, and ${\bf n} = [ {\bf n}_{1}^{\rm T},\ldots,{\bf n}_{L}^{\rm T} ]^{\rm T} \in {\mathbb C}^{LN}$.
Based on all the collective channel estimates $\{ {\hat {\bf h}}_{k} = [ {\hat {\bf h}}_{k1}^{\rm T},\ldots,{\hat {\bf h}}_{kL}^{\rm T} ]^{\rm T} \in {\mathbb C}^{LN} :k=1,\ldots,K\}$, the CPU can select an arbitrary combining vector ${\bf v}_k \in {\mathbb C}^{L N}$ for UE $k$.
Consequently, the data transmitted by UEs can be estimated.

We cannot compute the exact ergodic capacity of this setup due to the imperfect channel knowledge.
However, we can rigorously analyze the performance by using a standard capacity lower bound \cite{bjornson2017massive} referred to as an achievable spectral SE. The following SE is achievable when using minimum mean-squared error (MMSE) channel estimation, which is introduced in Section \ref{subsec:estimation}.
\begin{lemm}\label{lemm:SE ul 1}
When the MMSE channel estimates are available, an achievable uplink SE of UE $k$ is
\begin{equation}\label{eq:SE ul 1}
  {\sf{SE}}_k^{( {\rm ul},1 )} =  \frac{\tau _u}{\tau _c} {\mathbb{E}}\left\{ {{{\log }_2}\left( {1 + {\sf{SINR}}_k^{( {\rm ul},1 )}} \right)} \right\}
\end{equation}
where the instantaneous effective signal-to-interference-and-noise ratio (SINR) is
\begin{equation}\label{eq:SINR ul 1}
  {\sf{SINR}}_k^{( {\rm ul},1 )} = \frac{{{p_k}{{\left| {{\bf{v}}_k^{\rm{H}}{{{\bf{\hat h}}}_k}} \right|}^2}}}{{\sum\limits_{i = 1,i \ne k}^K {{p_i}{{\left| {{\bf{v}}_k^{\rm{H}}{{{\bf{\hat h}}}_i}} \right|}^2}}  + {\bf{v}}_k^{\rm{H}}\left( {\sum\limits_{i = 1}^K {{p_i}{{\bf{C}}_i}}  + {\sigma _{{\rm{ul}}}^2}{{\bf{I}}_{LN}}} \right){{\bf{v}}_k}}},
\end{equation}
with ${\bf C}_i = {\rm diag}({\bf C}_{i1},\ldots,{\bf C}_{iL})$ where ${\bf C}_{il} = {\mathbb{E}}\{ ({\bf h}_{il} - {\bf{\hat h}}_{il})({\bf h}_{il} - {\bf{\hat h}}_{il})^{\rm H} \}$ is the error correlation matrix of ${\bf{\hat h}}_{il}$.
\end{lemm}
\begin{proof}
It follows from the proof of \cite[Theo. 4.1]{bjornson2017massive}.
\end{proof}

{\bf Distributed Combining} In the second level, each AP can preprocess its signal {by computing local estimates of the data and then passing them to the CPU} for final decoding.
With the local combining vector ${\bf v}_{kl}$, AP $l$ computes its local estimate of $s_k$ as
\begin{equation}\label{eq:signal estimate ul 2}
   {\check s}_{kl} \triangleq {\bf{v}}_{kl}^{\rm{H}}{\bf{y}}_l^{{\rm{ul}}} = {\bf{v}}_{kl}^{\rm{H}}{{\bf{h}}_{kl}}{s_k} + \sum\limits_{i = 1,i \ne k}^K {{\bf{v}}_{kl}^{\rm{H}}{{\bf{h}}_{il}}{s_i}}  + {\bf{v}}_{kl}^{\rm{H}}{{\bf{n}}_l}.
\end{equation}
Any combining vector can be adopted in the above expression \eqref{eq:signal estimate ul 2}.
Unlike the centralized combining, however, AP $l$ can only use its own local channel estimates to design ${\bf v}_{kl}$.

The local estimates $\{ {\check s}_{kl}:l=1,\ldots,L \}$ are then sent to the CPU where they are linearly combined using the weights $\{ a_{kl}:l=1,\dots,L \}$ to obtain ${{\hat s}_k} = \sum\nolimits_{l = 1}^L {a_{kl}^*{{\check s}_{kl}}}$
which is eventually used to decode $s_k$, as
\begin{align}\label{eq:}\notag
{{\hat s}_k} &= \sum\limits_{l = 1}^L {a_{kl}^*{\bf{v}}_{kl}^{\rm{H}}{{\bf{h}}_{kl}}} {s_k} + \sum\limits_{l = 1}^L {a_{kl}^*\sum\limits_{i = 1,i \ne k}^K {{\bf{v}}_{kl}^{\rm{H}}{{\bf{h}}_{il}}{s_i}} }  + {{{\bf{n}}}'_k}\\
 &= {\bf{a}}_k^{\rm{H}}{{\bf{f}}_{kk}}{s_k} + \sum\limits_{i = 1,i \ne k}^K {{\bf{a}}_k^{\rm{H}}{{\bf{f}}_{ki}}{s_i}}  + {{{\bf{n}}}'_k}
\end{align}
where ${\bf f}_{ki} = [ {\bf{v}}_{k1}^{\rm{H}}{{\bf{h}}_{i1}} \ldots {\bf{v}}_{iL}^{\rm{H}}{{\bf{h}}_{kL}} ]^{\rm T} \in {\mathbb C}^{L}$ is the receive-combined channel vector between UE $k$ and each of the APs, ${\bf a}_{k} = [ a_{k1},\ldots,a_{kL} ]^{\rm T} \in {\mathbb C}^L$ is the weighting vector, ${{{\bf{n}}}'_k} = \sum\nolimits_{l = 1}^L {a_{kl}^*{\bf{v}}_{kl}^{\rm{H}}{{\bf{n}}_l}} $ is the effective noise, and $\{{{\bf{a}}_k^{\rm{H}}{{\bf{f}}_{ki}}} : i=1,\ldots,K\}$ is the effective channel.

Since the CPU does not have the knowledge of the effective channel ${\bf{a}}_k^{\rm{H}}{{\bf{f}}_{kk}}$, we utilize the well-considered \emph{use-and-then-forget} (UatF) capacity bound \cite[Theo. 4.4]{bjornson2017massive}, where we use the channel estimates for combining and then effectively ``forget" them before the signal detection, to obtain the achievable SE.

\begin{lemm}\label{lemm:SE ul 2}
An achievable uplink SE of UE $k$ is
\begin{equation}\label{eq:SE ul 2}
  {\sf{SE}}_k^{( {\rm ul},2 )} = \frac{{{\tau _u}}}{{{\tau _c}}}{\log _2}\left( {1 + {\sf{SINR}}_k^{( {\rm ul},2 )}} \right)
\end{equation}
with the effective SINR given by
\begin{align}\label{eq:SINR ul 2}\notag
  &{\sf{SINR}}_k^{( {\rm ul},2 )} \\&= \frac{{{p_k}{{\left| {{\bf{a}}_k^{\rm{H}}{\mathbb{E}}\left\{ {{{\bf{f}}_{kk}}} \right\}} \right|}^2}}}{{\sum\limits_{i = 1}^K {{p_i}{\mathbb{E}}\left\{ {{{\left| {{\bf{a}}_k^{\rm{H}}{{\bf{f}}_{ki}}} \right|}^2}} \right\}}  + {p_k}{{\left| {{\bf{a}}_k^{\rm{H}}{\mathbb{E}}\left\{ {{{\bf{f}}_{kk}}} \right\}} \right|}^2} + {\sigma _{{\rm{ul}}}^2}{\bf{a}}_k^{\rm{H}}{{\bf{\Lambda }}_k}{{\bf{a}}_k}}}
\end{align}
where ${\bf \Lambda}_k = {\rm diag}({\mathbb E}\{ \| {\bf v}_{k1}\|^2\},\ldots,{\mathbb E}\{ \| {\bf v}_{kL}\|^2\} ) \in {\mathbb C}^{L \times L}$ and the expectations are with respect to the channel estimates.
\end{lemm}
\begin{proof}
The proof is given in \cite[Appe. A]{[162]}.
\end{proof}

The achievable SE above holds for any combining scheme.
Unlike the achievable SE in Lemma \ref{lemm:SE ul 1}, it holds for any channel estimator (not only for the MMSE estimator).
The drawback with this bound is that it is only tight when ${\bf{a}}_k^{\rm{H}}{{\bf{f}}_{kk}}$ is close to its
mean value ${\bf{a}}_k^{\rm{H}}{\mathbb E\{{\bf{f}}_{kk}}\}$, but this seems to be the case in many mMIMO setups \cite{cellfreebook}.

The structure of \eqref{eq:SINR ul 2} allows computing the deterministic
weighting vector ${\bf a}_k$ that maximizes ${\sf{SINR}}_k^{( {\rm ul},2 )}$ as follows.
\begin{coro}
The effective SINR in \eqref{eq:SINR ul 2} for UE $k$ is maximized by
\begin{equation}\label{eq:weight level3}
  {{\bf{a}}_k} = {\left( {\sum\limits_{i = 1}^K {{p_i}{\mathbb{E}}\left\{ {{{\bf{f}}_{ki}}{\bf{f}}_{ki}^{\rm{H}}} \right\}}  + {\sigma _{{\rm{ul}}}^2}{{\bf{\Lambda }}_k}} \right)^{ - 1}}{\mathbb{E}}\left\{ {{{\bf{f}}_{kk}}} \right\}
\end{equation}
which leads to the maximum value
\begin{align}\label{eq:}\notag
{\sf{SINR}}_k^{( {\rm ul},2 )} &= {p_k}{\mathbb{E}}\left\{ {{{\bf{f}}_{kk}}} \right\}\left( \sum\limits_{i = 1}^K {{p_i}{\mathbb{E}}\left\{ {{{\bf{f}}_{ki}}{\bf{f}}_{ki}^{\rm{H}}} \right\}}  + {\sigma _{{\rm{ul}}}^2}{{\bf{\Lambda }}_k}\right. \\
&\left. - {p_k}{\mathbb{E}}\left\{ {{{\bf{f}}_{kk}}} \right\}{\mathbb{E}}\left\{ {{\bf{f}}_{kk}^{\rm{H}}} \right\} \right)^{ - 1}{\mathbb{E}}\left\{ {{{\bf{f}}_{kk}}} \right\}
\end{align}
\end{coro}
\begin{proof}
It follows from \cite[Lemm. B.10]{bjornson2017massive}. 
\end{proof}

The aforementioned approach is the so-called \emph{large-scale fading decoding (LSFD)} \cite{[5]}.
Although this approach offers the highest SE among schemes with local combining at each AP, it requires sharing of statistical information sharing among the APs to design the optimized LSFD weights.
{Alternatively, the weight $a_{kl}$ can be locally designed at each AP, such as
$
a_{kl}= \beta_{kl}^{\nu}
$
with different exponents $\nu$}, where ${\beta _{kl}} \triangleq {\rm{tr}}\left( {{{\bf{R}}_{kl}}} \right)/N$ is the large-scale fading coefficient that describes pathloss and shadowing and ${\bf{R}}_{kl}$ is the spatial correlation matrix that describes the spatial property of the channel.
When $\nu = 0$, the CPU creates its estimate of the signal $s_k$ from UE $k$ by
simply taking the average of the local estimates, as proposed in the early papers on CF mMIMO \cite{[19],[18]}.

Note that the value of the uplink SE of UE $k$ in CF mMIMO systems depends on the UE's combining vector ${\bf v}_k$.
The choices of the combining vector in the CF mMIMO literature can be found in Section \ref{subsec:combining precoding}.


\subsubsection{Downlink Data Transmission}

Let ${\bf w}_{il} \in {\mathbb C}^N$ denote the precoder that AP $l$ assigns to UE $i$.
During downlink data transmission, the received signal at UE $k$ is
\begin{equation}\label{eq:downlink received signal}
  y_k^{{\rm{dl}}} = \sum\limits_{l = 1}^L {{\bf{h}}_{kl}^{\rm{H}}\sum\limits_{i = 1}^K {{{\bf{w}}_{il}}{\varsigma _i}} }  + {n_k} = {\bf{h}}_k^{\rm{H}}\sum\limits_{i = 1}^K {{{\bf{w}}_i}{\varsigma _i}}  + {n_k},
\end{equation}
where ${\varsigma _i} \in {\mathbb C}$ is the independent unit-power data signal intended for UE $i$ (i.e., ${\mathbb{E}}\{ {{{\left\| {{\varsigma _i}} \right\|}^2}} \} = 1$), ${\bf w}_k = [ {\bf w}_{k1}^{\rm T},\ldots,{\bf w}_{kL}^{\rm T} ]^{\rm T} \in {\mathbb C}^{LN}$ is the collective precoding vector, and ${{n}_k} \sim {\cal N}_{\mathbb C}\left( { 0},\sigma_{\rm dl}^2 \right)$ is the receiver noise.
Normally, the collective precoding vector ${\bf w}_{i}$ is presented as
\begin{equation}\label{eq:}
  {\bf w}_{i} = {\sqrt {\rho_i}} {\bar {\bf w}}_i,
\end{equation}
where ${\bar {\bf w}}_i$ determines the spatial directivity of the transmission and satisfies ${\mathbb E}\{\| {\bar {\bf w}}_i \|^2\} = 1$ such that $\rho_i \ge 0$ is the transmit power allocated to UE $i$.

For a specific choice of precoding, the \emph{hardening bound} is used to compute the downlink SE.

\begin{lemm}\label{lemm:SE dl}
An achievable downlink SE of UE $k$ is
\begin{equation}\label{eq:SE dl}
  {\sf{SE}}_k^{( {\rm dl} )} = \frac{{{\tau _d}}}{{{\tau _c}}}{\log _2}\left( {1 + {\sf{SINR}}_k^{( {\rm dl} )}} \right)
\end{equation}
with the effective SINR given by
\begin{equation}\label{eq:SINR dl}
  {\sf{SINR}}_k^{( {\rm dl} )} = \frac{{{\rho_k}{{\left| {{\mathbb{E}}\left\{ {\bf h}_k^{\rm H} {\bar {\bf w}}_k \right\}} \right|}^2}}}{{\sum\limits_{i = 1}^K {{\rho_i}{\mathbb{E}}\{ {{{\left| {\bf h}_k^{\rm H} {\bar {\bf w}}_i \right|}^2}} \}}  - {\rho_k}{{\left| {\bf h}_k^{\rm H} {\bar {\bf w}}_k \right|}^2} + {\sigma _{{\rm{dl}}}^2}}}
\end{equation}
and the expectation is with respect to the channel realizations.
\end{lemm}
\begin{proof}
The proof is given in \cite[Appe. C.3.6]{bjornson2017massive}.
\end{proof}

In contrast to the uplink SEs of UE $k$ that only depend on the UE's combining vector, the downlink SE depends on the precoding vectors of all UEs, i.e., $\{{ {\bf w}}_i:i=1,\ldots,K\}$.
Consequently, the precoding vectors should be optimized jointly for all UEs instead of on a per-UE basis.
Alternatively, one can utilize the following uplink-downlink duality result to obtain a good heuristic solution.

\begin{lemm}\label{lemm:uplink downlink deality}
Let $\{{ {\bf v}}_i:i=1,\ldots,K\}$ and $\{ p_i:i=1,\ldots,K \}$ denote by the set of combining vectors and transmit powers used in the uplink.
If the normalized precoding vectors are selected as
\begin{equation}\label{eq:uplink downlink deality}
  {\bar {\bf w}}_i = \frac{{\bf v}_i}{\sqrt {{\mathbb E}\{\|{\bf v}_i\|^2\}}},
\end{equation}
then there exists a downlink power control policy $\rho_i : \forall i$ with $\sum\nolimits_{i=1}^K {\rho_i / \sigma_{\rm dl}^2} = \sum\nolimits_{i=1}^K {p_i / \sigma_{\rm ul}^2}$ for which
\begin{equation}\label{eq:}
  {\sf{SINR}}_k^{( {\rm dl} )} = {\sf{SINR}}_k^{( {\rm ul,2})}, \forall k.
\end{equation}
where ${\sf{SINR}}_k^{( {\rm dl} )}$ is the effective SINR of UE $k$ in the downlink and ${\sf{SINR}}_k^{( {\rm ul})}$ is the effective SINR of UE $k$ in the uplink with distributed cooperation manner.
\end{lemm}
\begin{proof}
The proof is given in \cite[Appe.]{[163]}.
\end{proof}

The above lemma implies that the downlink precoders in CF mMIMO networks can be selected based on the uplink combiners as in \eqref{eq:uplink downlink deality}.
Consequently, an achievable downlink SE for UE $k$ can be achieved
by properly selecting the power control coefficients $\{\rho_i : \forall i\}$ and normalized precoding vectors $\{{\bar {\bf w}}_i:\forall i\}$.
Similar to that in the uplink, we consider two levels of cooperation among APs for precoding design as shown in Fig.~\ref{fig:procedure}.
At both levels, we assume that the APs delegate the task of downlink data encoding to the CPU.

{\bf Centralized Precoding} In the first level, the CPU uses the uplink channel estimates to compute the normalized precoding vectors $\{{\bar {\bf w}}_{il}\}$ by exploiting channel reciprocity.
Motivated by the uplink-downlink duality, we select the downlink precoding vectors according to \eqref{eq:uplink downlink deality}.
Once the precoding vectors are computed, they are used by the CPU to form the downlink signal of any given AP $l$, as
\begin{equation}\label{eq:}
  {\bf x}^{\rm{dl}}_l = \sum\limits_{l = 1}^K {{\sqrt {\rho_i}} {\bar {\bf w}}_{il}} {\varsigma _{{i}} },
\end{equation}
which is sent to the AP via the fronthaul link for transmission.

{\bf Distributed Precoding} In the second level, AP $l$ can locally select the precoding vector ${\bar {\bf w}}_{il}$ on the basis of its local channel estimates $\{{{\bf{\hat h}}}_{il} \}$ instead of delegating the task to the CPU.
In this case, only the downlink data signals $\{ \varsigma _{i} \}$ are sent from the CPU to AP $l$ in each coherence block.

The choices of the precoding vector in the CF mMIMO literature can be found in Section \ref{subsec:combining precoding}.

\subsection{Channel Model}\label{subsec:channel model}

Line-of-sight (LoS) channels are widely considered to characterize propagation channels, which generally contain many propagation paths; one is the direct path, and the others are paths where the signals are scattered on different objects.
The direct path is typically referred to as the LoS component, and the scattered paths are referred to as the Non-LoS (NLoS) component.
The interaction between these paths leads to fading phenomena, which is often modeled statistically using \emph{Rician} fading (sometimes written as Ricean fading).
The main assumption is that the complex-valued channel coefficient between UE $k$ and AP $l$ in the complex baseband can be divided into two parts \cite{[120]}:
\begin{equation}\label{eq:uncorrelated rician with phase shift}
   {h_{kl}} = {{\bar h}_{kl}}{e^{j{\varphi _{kl}}}} + {g_{kl}},
\end{equation}
where ${{\bar h}_{kl}} \ge 0$ is the magnitude of the LoS component between UE $k$ and AP $l$ and ${\varphi _{kl}} \in [0, 2\pi)$ is the corresponding phase-shift. The second part, ${g_{kl}}$, represents the NLoS component comprising all the scattered paths, of which each is of roughly the same strength but substantially weaker than the LoS component (this is why it needs to be modeled separately).
Motivated by the central limit theorem, ${g_{kl}}$ is modeled by a Gaussian distribution, which implies ${g_{kl}} \sim {{\cal N}_{\mathbb{C}}}( {0,{\beta _{kl}}} )$, where ${\beta _{kl}}\geq 0$ is the variance.
This is called Rayleigh fading since $\left|g_{kl}\right|$ is Rayleigh distributed, i.e., $\left|g_{kl}\right| \sim {\rm{Rayleigh}}( {\sqrt{{\beta _{kl}}/2}} )$.
Under these assumptions, the magnitude $\left| {{h_{kl}}} \right|$ of the channel coefficient is Rice distributed, i.e., $\left| {{h_{kl}}} \right| \sim {\rm{Rice}}( {{{\bar h}_{kl}},\sqrt{{\beta _{kl}}/2}} )$, which is why it is called Rician fading.

When the channel is assumed to be perfectly known at the receiver, the phase-shift ${\varphi _{kl}}$ will not affect the communication performance since the receiver can compensate for it.
Hence, it is common to omit ${\varphi _{kl}}$ in the performance analysis of Rician fading channels.
Consequently, $h_{kl} $ in \eqref{eq:uncorrelated rician with phase shift} can be drawn as ${h_{kl}} \sim {{\cal N}_{\mathbb{C}}}\left( {{{\bar h}_{kl}},{\beta _{kl}}} \right)$.
However, we cannot neglect the phase ${\varphi _{kl}}$ when analyzing practical systems where the receiver needs to estimate the channel, since the value of ${\varphi _{kl}}$ varies at the same pace as ${g_{kl}}$ and ${\varphi _{kl}}$ affects the strong LoS component where the impact of ${\varphi _{kl}}$ cannot be ignored.
However, the results obtained with a perfectly-known ${\varphi _{kl}}$ can be interpreted as an upper bound on what is practically achievable.

Note that \eqref{eq:uncorrelated rician with phase shift} represents the channel when single-antenna UEs and single-antenna APs are considered.
In the case of single-antenna UEs and $N$-correlated-antenna APs, the channel between UE $k$ and AP $l$ is no longer a scalar but a $N$-dimensional vector as \cite{wang2020uplink}
\begin{equation}\label{eq:Rician}
  {{\bf{h}}_{kl}} = {e^{j{\varphi _{kl}}}} {{\bf{\bar h}}_{kl}} + {{\bf{ g}}_{kl}}
\end{equation}
where ${{{\bf{\bar h}}}_{kl}}$ and ${{\bf{g}}_{kl}} \sim {{\cal N}_{\mathbb{C}}}\left( {{\bf 0},{{\bf{R}}_{kl}}} \right)$ represent the LoS and NLoS component, respectively, and $\varphi _{kl}$ is the common phase shift.
If the phase-shift ${\varphi}_{kl}$ is neglected, the channel ${{\bf{h}}_{kl}}$ can be viewed as a realization of the circularly symmetric complex Gaussian distribution  \cite{[187],[222]}
\begin{equation}\label{eq:}
  {{\bf{h}}_{kl}} \sim {{\cal N}_{\mathbb{C}}}\left( {{{{\bf{\bar h}}}_{kl}},{{\bf{R}}_{kl}}} \right).
\end{equation}

Rayleigh fading is a tractable model for rich scattering scenarios without an LoS path, where the AP antenna array is surrounded by many scattering objects, as compared to the number of antennas per AP.
Rayleigh fading channel is widely used to describe the basic properties of wireless propagation.
Thus, the channel response $h_{kl}$ is distributed as ${{h}}_{kl}  \sim {\cal{N}}_{\mathbb{C}}\left( {{{0}},{\beta}_{kl}}\right)$ and the multi-antenna channel ${\bf h}_{kl}$  is distributed as
\begin{equation} \label{eq:Rayleigh}
{\bf{h}}_{kl}  \sim {\cal{N}}_{\mathbb{C}}\left( {{\bf{0}},{\bf{R}}_{kl}}\right).
\end{equation}

\subsection{Duplex Protocol}\label{subsec:duplex protocol}

According to whether the uplink and downlink are separated in time or frequency, a CF mMIMO system can operate in TDD or FDD mode.
In TDD mode, the signaling overhead scales with the number of the served UEs but is independent of the number of AP antennas due to the \emph{channel reciprocity} that appears when transmitting in both directions in the same band. This means that one can perform the downlink precoding based on the CSI obtained from the uplink pilots since the channel response is the same in both directions.
However, the channel reciprocity is not available in FDD since the uplink and downlink channels are in different bands, which introduces additional CSI acquisition and feedback overhead that is not only related to the number of served UEs but also the number of the AP antennas.
Most of the works on CF mMIMO systems assume TDD mode to avoid the exchange of CSI and precoding/combining vectors \cite{[18],[19],[89],[163]}.

However, in practice, this is not a design choice but rather determined by the spectrum license available for the system. Hence, it is also important to develop FDD-based CF mMIMO systems \cite{[150],[242]}.
To limit the CSI acquisition and feedback overhead in FDD-based systems, \cite{[150]} and \cite{[242]}  exploit the property of the so-called \emph{angle reciprocity}, which means the angles-of-departure (AoDs) are similar in both uplink and downlink. When the propagation channels are sufficiently sparse to utilize this property, the required overhead scales only with the number of the served UEs.
In \cite{[150]}, the authors exploited the discrete Fourier transform operation and log-likelihood function to estimate the multipath component for the angle-of-arrival (AoA) and large-scale fading coefficients.
Based on estimated AoA, linear precoding/combining schemes were proposed with only scales with the number of the served UEs.
In \cite{[242]}, the authors proposed a path gain information feedback scheme where the required overhead for the channel vector quantization scales linearly with the number of dominating paths instead of the number of the serving antennas.

\subsection{Scalability Issues}\label{subsec:scalable}

Scalability is an essential issue for network technology to be practically implemented, particularly when designing a technology where a large number of APs are supposed to cooperate.
According to  \cite{[163]}, a network is scalable if all the following tasks for the APs have finite complexity and resource requirement when the number of UEs tends to infinity:
\begin{enumerate}
  \item Signal processing for channel estimation;
  \item Signal processing for data reception and transmission;
  \item Fronthaul signaling for data and CSI sharing;
  \item Power control optimization,
\end{enumerate}

The naive form of CF mMIMO in which each AP is required to process and share the data signals related to all UEs fails to be scalable since the computational complexity and fronthaul load associated with the above-listed tasks grow linearly (or faster) with the number of the UEs.
The user-centric approach (also referred to as dynamic cooperating clustering \cite{bjornson2011optimality,Bjornson2013d}) takes the first step towards scalability by letting each AP only be responsible for a limited number of UEs in ${\cal D}_l \subset \left\{{ 1,\ldots,K }\right\}$, $l = 1,\ldots,L$, instead of all of them \cite{[11]}.
This makes a CF mMIMO system meet the first three conditions listed above if the cardinality $|{\cal D}_l|$ is constant as $K \to \infty $ for $l = 1,\ldots,L$.
The reason comes from the fact that AP $l$ only needs to compute the channel estimates and combining/precoding vectors for $|{\cal D}_l|$ UEs with a constant complexity as $K  \to \infty $.
Moreover, AP $l$ only needs to receive/send data related to these $|{\cal D}_l|$ UEs via the fronthaul network, which is a constant number as $K  \to \infty $.
How to select the UE sets ${\cal D}_l$ for $l = 1,\ldots,L$ in a scalable way while ensuring the service to all UEs is elaborated in Section \ref{subsec:access}.
Suboptimal power control policies, such as fractional power control \cite{[119],[166]}, are needed to limit the complexity of power control.

The signal processing procedure of the scalable CF mMIMO shares the similar methodology and mathematical expressions of the original alternative introduced earlier in this section by only letting the uplink combining vector ${\bf v}_{kl} = {\bf 0}$ and the downlink precoding vector ${\bf w}_{kl} = {\bf 0}$ for $k \notin {\cal D}_l$, $l = 1,\ldots,L$.
The design of scalable combining and precoding vectors can be found in Section \ref{subsec:combining precoding}.
Moreover, from the perspective of the UEs we also denote by ${\cal M}_k \subset \left\{{ 1,\ldots,L }\right\}$, $k = 1,\ldots,K$, the subset of APs serving UE $k$.

\subsection{Channel Hardening and Favorable Propagation}
%
%
%

Channel hardening and favorable propagation are the two basic virtues of cellular mMIMO \cite[Sec. 2.5]{bjornson2017massive}.
To be specific, \emph{channel hardening} makes the fading channel between an AP and a UE behave as almost deterministic after the precoding/combining has been applied, when the number of the serving antennas, i.e., $LN$ grows large.
Mathematically, this effect can be expressed as
\begin{equation}\label{eq:}
  \frac{\|{\bf h}_k\|^2}{{\mathbb E}\{\|{\bf h}_k\|^2\}}  \to  1 \quad {\text {as}} \quad LN \to \infty.
\end{equation}
This is the ultimate form of spatial diversity, which removes the impact of small-scale fading. In practice, a small amount of hardening is sufficient to alleviate the worst effects of fading and enable the resource allocation to be based on long-term statistics instead of the small-scale fading variations.

Beyond that, the directions of two UE channels are asymptotically orthogonal when the number of antennas approaches infinity, leading to the so-called \emph{favorable propagation}, which is expressed as
 \begin{equation}\label{eq:}
  \frac{{{\bf h}_k^{\rm H}}{{\bf h}_i}}{\sqrt{{\mathbb E}\{\|{\bf h}_k\|^2\}{\mathbb E}\{\|{\bf h}_i\|^2\}}}  \to  0 \quad {\text {when}} \quad LN \to \infty, \ k \ne i.
\end{equation}
With approximate favorable propagation, one can get away with relatively simple signal processing techniques since interference vanishes automatically.
However, some interference suppression is generally preferred.

\section{Resource Allocation and Signal Processing}\label{sec:processing}

Well-designed schemes and algorithms for resource allocation and signal processing are the keys to boosting the system performance of the CF mMIMO networks.
This section will provide a comprehensive survey of different categories of resource allocation and signal processing schemes, including the ones for channel estimation, combining and precoding, user access and association, and power control.

\subsection{Channel Estimation}\label{subsec:estimation}

The main benefits of serving a UE through multiple APs materialize when the APs have CSI so the received uplink signals can be coherently combined in the joint processing, and the downlink transmissions can be precoded to combine over the air coherently.
Since the channels are time-varying, assuming that complete and perfect CSI is available at the APs and CPUs is not realistic.
Consequently, developing accurate and resource-efficient channel estimation techniques is vital to achieve good performance and, particularly, improve performance over legacy technologies.
In this subsection, we provide an overview of channel estimation techniques for CF mMIMO systems.

\subsubsection{Normalized Mean-Squared Error}

Before we elaborate different estimators, we first introduce the \emph{mean-squared error (MSE)} which indicates the ``distance'' from the estimated channel and the actual channel, i.e., ${{\mathbb{E}}\{ {{\| {{{\bf{h}}} - {{{\bf{\hat h}}}}} \|} ^2}\} }$, where ${\bf{h}}$ and ${\bf{\hat h}}$ denote an arbitrary channel response and its estimate, respectively.
However, the value of the MSE depends on the average channel gain, and hence a strong channel might have larger errors in absolute terms than a weaker one.
To reasonably quantify the accuracy of an estimator, we consider the relative size of the error, i.e., the \emph{normalized MSE (NMSE)}.
The NMSE between AP $l$ and UE $k$ using an arbitrary estimator is represented as
\begin{equation}\label{eq:}
  {\sf{NMSE}}_{kl} = \frac{{{\mathbb{E}}\{ {{\| { {{\bf{\tilde h}}_{kl}}} \|} ^2}\} }}{{{\mathbb{E}}\{ {{\| {{{\bf{h}}_{kl}}} \|}^2}\} }} = \frac{{{\rm{tr}}({{\bf{C}}_{kl}})}}{{{\rm{tr}}({{\bf{R}}_{kl}})}}
\end{equation}
where ${\bf{\tilde h}}_{kl} = {\bf{h}}_{kl} - {\bf{\hat h}}_{kl}$ and ${\bf C}_{kl} = {\mathbb{E}}\{ {{\bf{\tilde h}}_{kl} {{\bf{\tilde h}}_{kl}^{\rm H}} } \}$ denote the estimation error of the considered estimator and its corresponding error correlation matrix, respectively.

When it comes to the collective channel of UE $k$, i.e., ${\bf h}_k$, the NMSE of its estimate ${\bf {\hat h}}_k$ can be computed as \cite{cellfreebook}
\begin{equation}\label{eq:NMSE_k}
  {\sf{NMSE}}_k  = \frac{{{\mathbb{E}}\{ {{\| {{\bf D}_k {{\bf{\tilde h}}_{k}}} \|} ^2}\} }}{{{\mathbb{E}}\{ {{\| {{\bf D}_k {{\bf{h}}_{k}}} \|}^2}\} }} = \frac{{\sum\nolimits_{l=1}^L {{\rm{tr}}({\bf D}_{kl}{{\bf{C}}_{kl}})} }}{{\sum\nolimits_{l=1}^L {{\rm{tr}}({\bf D}_{kl}{{\bf{R}}_{kl}})} }},
\end{equation}
where $ {\bf {\tilde  h}}_k = [ {\bf {\tilde h}}_{k1}^{\rm T},\ldots,{\bf {\tilde h}}_{kL}^{\rm T} ]^{\rm T} \in {\mathbb C}^{LN}$ and block-diagonal matrix ${\bf D}_{k}= {\rm{diag}}\left({{\bf D}_{k1} ,\ldots, {\bf D}_{kL}}\right) $ with ${\bf D}_{kl} = {\bf I}_N$ if $l \in {\cal M}_k$ and ${\bf D}_{kl} = {\bf 0}_N$ otherwise.
Note that \eqref{eq:NMSE_k} is not the sum or average of the individual NMSEs between UE $K$ and its serving AP in ${\cal M}_k$, but contains a summation of MSEs in the numerator normalized by a summation of channel gains. Hence, it is the APs with strong channels that dominate in the summation.

%
%

\subsubsection{Pilot-based Channel Estimation}

The most commonly used approach for CSI acquisition is by transmitting uplink pilot signals, where a predefined pilot signal is transmitted from the UE-side antenna, and all the antenna at the APs can simultaneously receive the transmission and compare it with the known pilot signal to estimate the channel response from the transmitting antenna.
Suppose we instead need to estimate the channel response from two transmitting antennas. In that case, two orthogonal pilot signals are generally required to separate the signals from the two antennas (unless there is other prior information that allows for separation).
The number of orthogonal pilot signals is proportional to the number of transmit antennas, while any number of the receive antennas can ``listen" to the pilots simultaneously and estimate their respective channels to the transmitters.
When the UEs transmit their pilots, the received signal can be represented as ${\bf{y}}_{{t_k}l}^{{\rm{pilot}}}$ which is given in (\ref{eq:pilot}), and based on that there exist several different channel estimators.

{\bf Least-squares Estimator} If the statistics are unknown or unreliable, it might be necessary to consider estimators that require no prior statistical information.
The least-square (LS) estimator has been used for this purpose.
{The LS estimator minimizes ${\| {{{\bf{y}}_{{t_k}l}^{{\rm{pilot}}}} - \sqrt {{p_k}} {\tau _p}{{{\bf{\hat h}}}_{kl}}} \|^2}$, which is achieved by \cite[Sec. 3.4]{bjornson2017massive}}
\begin{equation}\label{eq:ls estimate}
{\bf{\hat h}}_{kl}^{{\rm{LS}}} = \frac{1}{{\sqrt {{\tau _p}{p_k}} }}{\bf{y}}_{{t_k}l}^{{\rm{pilot}}} = {{\bf{h}}_{kl}} + \sum\limits_{i \in {{\cal S}_k}/\left\{ k \right\}} {\frac{{\sqrt {{p_i}} }}{{\sqrt {{p_k}} }}{{\bf{h}}_{il}}}  + \frac{1}{{\sqrt {{\tau _p}{p_k}} }}{{\bf{n}}_{{t_k}l}}.
\end{equation}
{If the channels are Rayleigh fading as defined in \eqref{eq:Rayleigh},
%
substituting \eqref{eq:ls estimate} into \eqref{eq:NMSE_k},
the NMSE of the LS estimator can be computed as
\begin{equation}\label{eq:}
  {\sf{NMSE}}_k^{{\rm{LS}}} = \frac{{\sum\nolimits_{l \in {{\cal M}_k}} {{\rm{tr}}(\sum\nolimits_{i \in {{\cal S}_k}/\{ k\} } {\frac{{{p_i}}}{{{p_k}}}{{\bf{R}}_{il}}}  + \frac{{\sigma _{{\rm{ul}}}^2}}{{{p_k}{\tau _p}}}{{\bf{I}}_N})} }}{{\sum\nolimits_{l \in {{\cal M}_k}} {{\rm{tr}}({{\bf{R}}_{kl}})} }}.
\end{equation}
Note that the estimate ${\bf{\hat h}}_{kl}^{{\rm{LS}}}$ and the estimation error ${\bf{\tilde h}}_{kl}^{{\rm{LS}}}$ are correlated.}

{\bf MMSE-type Estimators} The MMSE estimator has this name because it minimizes the MSE, and thereby also the NMSE.
It takes different forms depending on the channel statistics. {If the channels are Rayleigh fading, then the MMSE estimate
 ${\bf{\hat h}}_{kl}^{{\rm{MMSE}}}$ can be achieved by \cite[Sec. 3.2]{bjornson2017massive}}
\begin{equation}\label{eq:MMSE estimate}
{\bf{\hat h}}_{kl}^{{\rm{MMSE}}}  = \sqrt {{\tau _p}{p_k}} \left({\mathbb{E}}\left\{ {{\bf{y}}_{{t_k}l}^{{\rm{pilot}}}{( {{\bf{y}}_{{t_k}l}^{{\rm{pilot}}}} )^{\rm{H}}}} \right\} \right)^{-1}{\bf{y}}_{{t_k}l}^{{\rm{pilot}}}.
\end{equation}
%
{Substituting \eqref{eq:MMSE estimate} into \eqref{eq:NMSE_k},
and the NMSE of the MMSE estimator can be computed as
\begin{equation}\label{eq:}
  {\sf{NMSE}}_k^{{\rm{MMSE}}} = \frac{{\sum\nolimits_{l \in {{\cal M}_k}} {{\rm{tr}}({{\bf{R}}_{kl}} - {\tau _p}{p_k}{{\bf{R}}_{kl}}{\bf{\Psi }}_{{t_k}l}^{ - 1}{{\bf{R}}_{kl}})} }}{{\sum\nolimits_{l \in {{\cal M}_k}} {{\rm{tr}}({{\bf{R}}_{kl}})} }}.
\end{equation}}

This is the optimal channel estimator from an MSE and NMSE perspective, thus all other estimators described in this survey will provide larger MSEs.
Recall that unlike the LS estimate, the MMSE estimate ${\bf{\hat h}}_{kl}^{{\rm{MMSE}}}$ and the estimation error ${\bf{\tilde h}}_{kl}^{{\rm{MMSE}}}$ are independent vectors.

Another alternative MMSE-type estimator is obtained by estimating each element of ${{\bf{h}}_{kl}}$ separately and thereby ignore these correlations between the elements \cite{Shariati2014a,[222]}.
More precisely, we can consider one of the $N$ elements in ${\bf{y}}_{{t_k}l}^{{\rm{pilot}}}$  at a time.
The resulting element-wise MMSE (EW-MMSE) estimator that estimate the $n$th element ${\left[ {{{\bf{h}}_{kl}}} \right]_n}$ is given as
\begin{equation}\label{eq:EWMMSE estimate}
{\left[ {{\bf{\hat h}}_{kl}^{{\rm{EW - MMSE}}}} \right]_n} = \frac{{\sqrt {{p_k}{\tau _p}} {{\left[ {{{\bf{R}}_{kl}}} \right]}_{nn}}}}{{\sum\limits_{i \in {{\cal S}_k}} {{p_i}{\tau _p}{{\left[ {{{\bf{R}}_{il}}} \right]}_{nn}}}  + {\sigma _{{\rm{ul}}}^2}}}{\left[ {\bf{y}}_{{t_k}l}^{{\rm{pilot}}} \right]_n}.
\end{equation}
{Since this is a type of MMSE estimator, the EW-MMSE estimate ${[ {{\bf{\hat h}}_{kl}^{{\rm{EW - MMSE}}}} ]_n}$ and the corresponding estimation error ${[ {{\bf{\tilde h}}_{kl}^{{\rm{EW - MMSE}}}} ]_n} $ are independent scalars.
%
Substituting \eqref{eq:EWMMSE estimate} into \eqref{eq:NMSE_k},
the NMSE of the EW-MMSE estimator can be computed as
\begin{equation}\label{eq:}
  {\sf{NMSE}}_k^{{\rm{EW-MMSE}}} = \frac{{\sum\nolimits_{l \in {{\cal M}_k}} {\sum\nolimits_{n = 1}^N { {{\left[ {{{\bf{R}}_{kl}}} \right]}_{nn}} - \frac{{{p_k}{\tau _p}{{\left( {{{\left[ {{{\bf{R}}_{kl}}} \right]}_{nn}}} \right)}^2}}}{{\sum\nolimits_{i \in {{\cal S}_k}} {{p_i}{\tau _p}{{\left[ {{{\bf{R}}_{il}}} \right]}_{nn}}}  + {\sigma _{{\rm{ul}}}^2}}} } } }}{{\sum\nolimits_{l \in {{\cal M}_k}} {{\rm{tr}}({{\bf{R}}_{kl}})} }}.
\end{equation}}

However, there can be cross-correlation between the estimate and estimation error for different antennas and its existence demonstrate the suboptimality of EW-MMSE; an optimal estimator exploits all correlation to lower the MSE.
The EW-MMSE estimator results in larger estimation errors since the correlation between the variables are not utilized to improve the estimation quality. One benefit of the EW-MMSE estimator is that it requires less statistical information since only the diagonals of the spatial correlation matrices are utilized. Moreover, it has lower computational complexity than the MMSE estimator, except in the particular case when all the spatial correlation matrices are diagonal so that one can estimate each channel element separately without a performance loss.

Since MMSE-type estimators depend on the channel statistics, they also take different forms when changing the channel model.
When considering Rician fading channels, as defined in \eqref{eq:Rician}, there are different estimators of the MMSE-type.
If the LoS component ${{{\bf{\bar h}}}_{kl}}$, the spatial
correlation matrix ${{\bf{R}}_{kl}}$, and the phase-shift ${\varphi _{kl}}$ are known, the phase-aware MMSE (PA-MMSE) estimate of ${{\bf{h}}_{kl}}$ is given as \cite{[120],wang2020uplink}
\begin{equation}
{\bf{\hat h}}_{kl}^{{\rm{PA - MMSE}}} = {{{\bf{\bar h}}}_{kl}}{e^{j{\varphi _{kl}}}} + \sqrt {\tau_p { p_k}} {{\bf{R}}_{kl}}{\bm{\Psi}} _{{t_k}l}^{ - 1}( {{\bf{y}}_{{t_k}l}^{{\rm{pilot}}} - {{{\bf{\bar z}}}_{{t_k}l}}} )
\end{equation}
where ${{{\bf{\bar z}}}_{{t_k}l}} = \sum\limits_{i \in {{\cal S}_k}} {\sqrt {{p_i}{\tau _p}} } {{{\bf{\bar h}}}_{il}}{e^{j{\varphi _{il}}}}$.
{The estimate ${\bf{\hat h}}_{kl}^{{\rm{PA - MMSE}}}$ and estimation error ${\bf{\tilde h}}_{kl}^{{\rm{PA - MMSE}}}$ are independent random variables.}

If the channel statistics ${{{\bf{\bar h}}}_{kl}}$, ${{\bf{R}}_{kl}}$ are available while the phase ${\varphi _{kl}}$ is unknown and uniformly distributed from $0$ to $2\pi$, then the linear MMSE (LMMSE) estimator of ${{\bf{h}}_{kl}}$ is \cite{[120]}
\begin{equation}
{\bf{\hat h}}_{kl}^{{\rm{LMMSE}}} = \sqrt {{p_k}} {{{\bf{R}}}'_{kl}}{\left( {{\bm{\Psi }}' _{{t_k}l}} \right)^{ - 1}}{\bf{y}}_{{t_k}l}^{{\rm{pilot}}},
\end{equation}
where ${\bf{R}'}_{kl} = {\bf{R}}_{kl} + {{{\bf{\bar h}}}_{kl}}{\bf{\bar h}}_{kl}^H$ and ${{\bm{\Psi}}' _{{t_k}l}} = \sum\nolimits_{i \in {{\cal S}_k}} {{p_i}{\tau _p}} \left( {{{\bf{R}}_{il}} + {{{\bf{\bar h}}}_{il}}{\bf{\bar h}}_{il}^H} \right) + {\sigma _{{\rm{ul}}}^2}{{\bf{I}}_N}$. {The LMMSE estimation ${\bf{\hat h}}_{kl}^{{\rm{LMMSE}}}$ and the estimation error ${\bf{\tilde h}}_{kl}^{{\rm{LMMSE}}}$ are independent random variables.}
\begin{figure}[t!]
\centering
\includegraphics[scale=0.64]{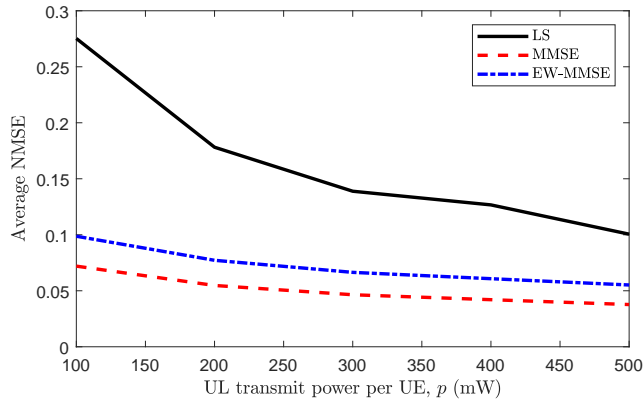}
\caption{Average NMSE versus $p$ with {MMSE \cite{bjornson2017massive}}, EW-MMSE \cite{Shariati2014a,[222]}, and {LS \cite{bjornson2017massive} estimator}.
\label{fig:estimation 1}}
\end{figure}

{It is worth noting that the aforementioned pilot-based schemes have the same names as the ones used in cellular mMIMO literature, because they are derived to minimize the same general metrics. However, the estimator expressions differ since pilot contamination affects the system differently, the notation is different, and we also notice that the CF mMIMO estimators can be computed separately at each AP (instead of centrally at one CPU as in cellular mMIMO) since the channel vectors are independent between APs.
Therefore, the estimation schemes are specially tailored for CF mMIMO systems.}
A comparison between the aforementioned estimators over spatially correlated Rayleigh fading channels is provided in Fig.~\ref{fig:estimation 1}, where the average NMSE of $K = 50$ UEs is shown as a decreasing function of the uplink transmit power per UE, $p$, with $L = 100$ APs and each equipped $N = 4$ antennas.
It can be seen that the MMSE-type estimators significantly outperform the LS estimator since the LS estimator has no prior information about the channel statistics.
As mentioned before, the EW-MMSE estimator results in a larger average NMSE compared to the MMSE estimator since the former only exploits partial statistical information while the latter relies on the full statistical channel knowledge.
The performance gap between the MMSE and EW-MMSE estimator will substantially decrease when considering the spatially correlated Rician fading channels since the existence of LoS components weakens the effect of spatial correlation, and the diagonals of the correlation matrix of these two MMSE-type estimators are identical.
Moreover, all aforementioned channel estimation schemes are summarized in Table \ref{tab:channel estimation}, where the abbreviation ``In." stands for whether the channel estimate and estimation error are independent, and the closed-form expressions can be found in the references.

\begin{table*}[tp]
  \centering
  \fontsize{9}{12}\selectfont
  \caption{Pilot Training-Based Channel Estimation.}
  \label{tab:channel estimation}
    \begin{tabular}{ !{\vrule width1.2pt}  m{1.2cm}<{\centering} !{\vrule width1.2pt}  m{3.2cm}<{\centering} !{\vrule width1.2pt}  m{4.6 cm}<{\centering} !{\vrule width1.2pt}  m{4.6 cm}<{\centering} !{\vrule width1.2pt}  m{0.6cm}<{\centering} !{\vrule width1.2pt}}

    \Xhline{1.2pt}
        \rowcolor{gray!50} \bf Scheme  & \bf Estimate of ${\bf h}_{kl}$  & \bf Mean and Covariance of Estimate & \bf Mean and Covariance of Estimation Error &  \bf Inde. \cr
    \Xhline{1.2pt}

        \multirow{2}{*}{ LS  } & \multirow{2}{*}{ ${\bf{\hat h}}_{kl}^{{\rm{LS}}} = \frac{1}{{\sqrt {{p_k}{\tau _p}} }}{\bf{y}}_{{t_k}l}^{{\rm{pilot}}}$}  & ${\mathbb{E}}\{ {{\bf{\hat h}}_{kl}^{{\rm{LS}}}} \} = {\bf{0}}$ & ${\mathbb{E}}\{ {{\bf{\tilde h}}_{kl}^{{\rm{LS}}}} \} = {\bf{0}}$  & \multirow{2}{*}{ $\times$ } \\
        \cline{3-4} & & $\quad {\mathbb C}\{ {{\bf{\hat h}}_{kl}^{{\rm{LS}}}} \} = \sum\nolimits_{i \in {{\cal S}_k}} {\frac{{{p_i}}}{{{p_k}}}{{\bf{R}}_{il}}}  + \frac{{{\sigma _{{\rm{ul}}}^2}}}{{{p_k}{\tau _p}}}{{\bf{I}}_N}$ & $ \quad {\mathbb C}\{ {{\bf{\hat h}}_{kl}^{{\rm{LS}}}} \} = {\mathbb C}\{ {{\bf{\hat h}}_{kl}^{{\rm{LS}}}} \}  - {\bf R}_{kl}$ &  \cr\hline

        \multirow{2}{*}{ MMSE } & \multirow{2}{*}{ \makecell*[l]{${\bf{\hat h}}_{kl}^{{\rm{MMSE}}}  $ \\\hspace{-0.8em}$= \sqrt {{\tau _p}{p_k}} {{\bf{R}}_{kl}}{\bf{\Psi }}_{{t_k}l}^{ - 1}{\bf{y}}_{{t_k}l}^{{\rm{pilot}}}$}}  & ${\mathbb{E}}\{ {{\bf{\hat h}}_{kl}^{{\rm{MMSE}}}} \} = {\bf{0}}$ & ${\mathbb{E}}\{ {{\bf{\tilde h}}_{kl}^{{\rm{MMSE}}}} \} = {\bf{0}}$ & \multirow{2}{*}{ \checkmark  }\\
        \cline{3-4} & & ${\mathbb C}\{ {{\bf{\hat h}}_{kl}^{{\rm{MMSE}}}} \} = {\tau _p}{p_k}{{\bf{R}}_{kl}}{\bf{\Psi }}_{{t_k}l}^{ - 1}{{\bf{R}}_{kl}}$ & $ {\mathbb C}\{ {{\bf{\tilde h}}_{kl}^{{\rm{MMSE}}}} \} = {{{\bf{R}}}_{kl}} - {\mathbb C}\{ {{\bf{\hat h}}_{kl}^{{\rm{MMSE}}}} \}$   & \cr\hline


        \multirow{2}{*}{ \makecell*[c]{EW-\\MMSE }} & \multirow{2}{*}{ \makecell*[l]{${[ {{\bf{\hat h}}_{kl}^{{\rm{EW - MMSE}}}} ]_n} $ \\\hspace{-0.8em}$= \frac{{\sqrt {{p_k}{\tau _p}} {{\left[ {{{\bf{R}}_{kl}}} \right]}_{nn}}}}{{\sum\nolimits_{i \in {{\cal S}_k}} {{p_i}{\tau _p}{{\left[ {{{\bf{R}}_{il}}} \right]}_{nn}}}  + {\sigma _{{\rm{ul}}}^2}}}$\\\hspace{-0.8em}$\times {[ {\bf{y}}_{{t_k}l}^{{\rm{pilot}}} ]_n}$}}  & ${\mathbb{E}}\{ {[ {{\bf{\hat h}}_{kl}^{{\rm{EW - MMSE}}}} ]_n} \} = {{0}}$ & ${\mathbb{E}}\{ {[ {{\bf{\tilde h}}_{kl}^{{\rm{EW - MMSE}}}} ]_n} \} = {{0}}$ & \multirow{2}{*}{ \checkmark  }\\
        \cline{3-4} & & ${\mathbb C}\{ {[ {{\bf{\hat h}}_{kl}^{{\rm{EW - MMSE}}}} ]_n} \} = \frac{{{p_k}{\tau _p}{{\left( {{{\left[ {{{\bf{R}}_{kl}}} \right]}_{nn}}} \right)}^2}}}{{\sum\limits_{i \in {{\cal S}_k}} {{p_i}{\tau _p}{{\left[ {{{\bf{R}}_{il}}} \right]}_{nn}}}  + {\sigma _{{\rm{ul}}}^2}}}$ & $ {\mathbb C}\{ {[ {{\bf{\tilde h}}_{kl}^{{\rm{EW - MMSE}}}} ]_n} \} = {\left[ {{{\bf{R}}_{kl}}} \right]_{nn}} - {\mathbb C}\{ {[ {{\bf{\hat h}}_{kl}^{{\rm{EW - MMSE}}}} ]_n} \}$  &  \cr\hline

        \multirow{2}{*}{ \makecell*[c]{PA-\\MMSE }} & \multirow{2}{*}{ \makecell*[l]{${\bf{\hat h}}_{kl}^{{\rm{PA - MMSE}}} $ \\\hspace{-0.8em}$= {{{\bf{\bar h}}}_{kl}}{e^{j{\varphi _{kl}}}} + \sqrt {{{ p}_k}} {{\bf{R}}_{kl}}$\\\hspace{-0.8em}$\times {\bm{\Psi}} _{{t_k}l}^{ - 1}( {{\bf{y}}_{{t_k}l}^{{\rm{pilot}}} - {{{\bf{\bar z}}}_{{t_k}l}}} )$}}  & ${\mathbb{E}}\{ {{\bf{\hat h}}_{kl}^{{\rm{PA - MMSE}}}|{\varphi _{kl}}} \} = {{{\bf{\bar h}}}_{kl}}{e^{j{\varphi _{kl}}}}$ & ${\mathbb{E}}\{ {{\bf{\tilde h}}_{kl}^{{\rm{PA - MMSE}}}} \} = {\bf 0}$ & \multirow{2}{*}{ \checkmark  }\\
        \cline{3-4} & & ${\mathbb C}\{ {{\bf{\hat h}}_{kl}^{{\rm{PA - MMSE}}}|{\varphi _{kl}}} \} = {p_k}{\tau _p}{{\bf{R}}_{kl}}{\bm{\Psi}} _{{t_k}l}^{ - 1}{{\bf{R}}_{kl}}$ & $ {\mathbb C}\{ {{\bf{\tilde h}}_{kl}^{{\rm{PA - MMSE}}}} \} = {{\bf{R}}_{kl}} - {p_k}{\tau _p}{{\bf{R}}_{kl}}{\bm{\Psi}} _{{t_k}l}^{ - 1}{{\bf{R}}_{kl}}$  &  \cr\hline

        \multirow{2}{*}{ LMMSE } & \multirow{2}{*}{ \makecell*[l]{${\bf{\hat h}}_{kl}^{{\rm{LMMSE}}} $ \\\hspace{-0.8em}$= \sqrt {{p_k}} {{{\bf{R}}}'_{kl}}{\left( {{\bm{\Psi }}' _{{t_k}l}} \right)^{ - 1}}{\bf{y}}_{{t_k}l}^{{\rm{pilot}}}$}}  & ${\mathbb{E}}\{ {{\bf{\hat h}}_{kl}^{{\rm{LMMSE}}}} \} = {\bf{0}}$ & ${\mathbb{E}}\{ {{\bf{\tilde h}}_{kl}^{{\rm{LMMSE}}}} \} = {\bf{0}}$ & \multirow{2}{*}{ \checkmark  }\\
        \cline{3-4} & & ${\mathbb C}\{ {{\bf{\hat h}}_{kl}^{{\rm{LMMSE}}}} \} = {p_k}{\tau _p}{{{\bf{R}}}'_{kl}}{\left( {{\bm{\Psi}}' _{{t_k}l}} \right)^{ - 1}}{{{\bf{R}}}'_{kl}}$ & $ {\mathbb C}\{ {{\bf{\tilde h}}_{kl}^{{\rm{LMMSE}}}} \} = {{{\bf{R}}}'_{kl}} - {\mathbb C}\{ {{\bf{\hat h}}_{kl}^{{\rm{LMMSE}}}} \}$ &   \cr

    \Xhline{1.2pt}
    \end{tabular}
  \vspace{0cm}
\end{table*}

\subsubsection{ML-based Channel Estimation}

Machine learning (ML) is a powerful tool for identifying structure and data and making decisions, and thus can particularly be utilized for reducing the computational complexity of known algorithms or learning mappings between known and unknown variables \cite{bjornson2020two}.
Although the MMSE estimator is optimal for Rayleigh and Rician fading channels, the price to pay is the high computational complexity from inverting a matrix (see \eqref{eq:MMSE estimate}), if each AP has many antennas.
This motivates the design of ML-based channel estimation which might reduce the online computational complexity by exploiting data-driven signal processing algorithms.
Besides,  practical wireless channels for future wireless communications (e.g., mmWave channels) can only be approximately described by Rayleigh and Rician fading, thus the MMSE-type estimators are not optimal in practice. Hence, an ML-based estimator can potentially improve the estimation quality.
In \cite{[107]}, the authors proposed a fast and flexible denoising convolutional neural network (FFDNet) for the channel estimation in CF mMIMO systems.
By introducing a noise level map as input sub-images, a single neural network can be used to handle different noise levels and reduce the waiting time for training and testing.
The convolutional blind denoising network (CBDNet) is developed to improve the blind denoising performance for real-world noisy images, thus it can be used for boosting the quality of channel estimation by regarding the channel matrix as an image.
CBDNet-based channel estimation for mmWave mMIMO systems was proposed in \cite{[188]}, where the sparsity feature of the mmWave channel was exploited to achieve notable performance gain with a wide range of SNRs and fast convergence.

\begin{figure}[t]
\centering
\includegraphics[scale=0.64]{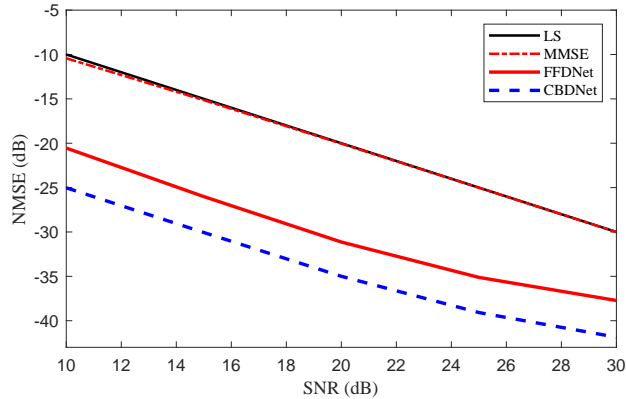}
\caption{NMSE versus SNR with CBDNet \cite{[188]}, FFDNet \cite{[107]}, {LS \cite{bjornson2017massive}}, and {MMSE \cite{bjornson2017massive}} estimator.
\label{fig:estimation 2}}
\end{figure}

Fig.~\ref{fig:estimation 2} compares the NMSE performance of CBDNet- and  FFDNet-based channel estimators with the conventional LS and MMSE estimators in a mmWave channel setup.
We first use fast Fourier transform-based modulation and demodulation schemes to expose the low-sparse information in the angle and delay domains.
To be specific, the ML-based methods learn the sparse virtual-channel coefficients information of the mmWave channel, with $\leq L$ scattering clusters having non-zero coefficients, and introduce the residual network to capture the low-dimensional sparse channel subspace that carries most of the power and propagation angle information efficiently.
It can be seen that the ML-based methods (i.e., FFDNet and CBDNet) significantly outperform the conventional methods (i.e., LS and MMSE) since the ML-based methods learn certain information about the propagation environment while the counterparts only assume covariance information that cannot fully capture the sparsity.
It is worth noting that the more prior information that one has about the propagation environment, the better the estimates can become.
In theory, one can develop a new MMSE estimator for the scenario simulated here, which will be optimal in terms of NMSE if the same side information learned by ML-based methods is available.
However, this prior information may be hard to describe mathematically and must be extracted from data; that is when the ML approaches (e.g., FFDNet and CBDNet) prove their abilities.
Moreover, it is illustrated that FFDNet has a larger NMSE than CBDNet since FFDNet only fits narrow noise levels and lacks the adaptability to out-of-trained range noise levels. While CBDNet can deal with whole and beyond training SNR range using continuous nonlinear joint loss function to enlarge the SNR range and achieve fast convergence because CBDNet converts the loss to the same order of magnitude.

\subsubsection{Other Methods}

Beyond the pilot- and ML-based channel estimation methods, other options exist for their specific applications in CF mMIMO systems.
Recall that pilot contamination is an inevitable phenomenon introduced by pilot reuse, which can substantially deteriorate the estimation accuracy and the system performance. It might be less of an issue in CF mMIMO than in conventional mMIMO due to the distributed nature, where each AP has relatively few antennas, but it cannot be neglected.
To mitigate pilot contamination, \cite{[152]} proposed a time-of-arrival (TOA)-based scheme that first estimates the TOA of the multipath channel and then filters out interfering signals of paths originating from distant APs.
It is shown that the proposed TOA-based scheme could outperform the LS estimator, given the TOA estimation is accurate.
The TOA-based channel estimation only depends on the currently received signal without the need for channel statistics; that is the key advantage of this scheme, which makes it more practically applicable.

In CF mMIMO systems, additional fronthaul overhead is required to exchange CSI between the APs and CPU, making blind channel estimation schemes attractive from an overhead perspective.
Independent component analysis (ICA), which is one of the blind source separation approaches, was employed in \cite{[183]} for blind channel estimation and signal detection in CF mMIMO systems.
The proposed scheme used ICA to separate and decode the received signals and to estimate the channels.
The estimated channel energy was used to differentiate the in-cell signals and the neighboring cell signals.
The reference bits were applied to identify the desired signal among signals within a cell.
However, this ICA-based method suffers from an error floor at high SINR due to inadequate ambiguity elimination.

When considering FDD mode, the key challenges are mainly CSI acquisition and feedback overhead.
Even if channel reciprocity does not hold, we can exploit angle reciprocity, in which the AoDs are similar in both uplink and downlink.
In \cite{[73]}, the authors proposed a discrete Fourier transform (DFT)-based channel estimation scheme in an FDD CF mMIMO system where the angle reciprocity of multipath components in both uplink and downlink is exploited.
Based on the DFT operation and log-likelihood function where the angle rotation is with a tiny amount of training overhead, {the required CSI requisition overhead} scales only with the number of served UEs.

\subsection{Receive Combining}\label{subsec:combining precoding}

In the uplink, the APs utilize the CSI that is acquired by channel estimation to perform the coherent signal processing, which is called \emph{receive combining}.
Receive combining is a linear projection that transforms the vector channels into effective scalar channels that support higher SEs than in the case where only one single-antenna AP serves each UE \cite{[241]}. The purpose of the receive combining is to make the desired signal much stronger than the sum of interfering signals, and noise \cite{[113],[49],[82]}, which requires CSI.
Different combining methods lead to substantially different SEs and computational complexities.
In this subsection, we comprehensively introduce the receive combining used in CF mMIMO networks with a centralized or distributed operation.

\subsubsection{MR Combining}

{The scheme with the lowest complexity is maximum ratio (MR) combining, defined as \cite[Sec. 4.1]{bjornson2017massive}}
\begin{equation}
{\bf{v}}_{kl}^{{\rm{MR}}} = {{{\bf{\hat h}}}_{kl}}.
\end{equation}
This is a vector that maximizes the ratio ${| {{\bf{v}}_{kl}^H{{{\bf{\hat h}}}_{kl}}}|^2}/{\left\| {{{\bf{v}}_{kl}}} \right\|^2}$ between the power of the desired signal and the squared norm of the combining vector \cite{[241]}. It shows that MR coherently combines all the received energy from the desired signal because the combining vector is matched to the channel response of the desired UE \cite{[105],[126],[79]/[80]}.

\subsubsection{MMSE-like Combining}

Although MR maximizes the gain of the desired signal, it might not be the preferable choice when there are interfering signals. MMSE-like methods can be utilized to identify a tradeoff between maximizing the signal gain and rejecting interference. There are different methods for distributed and centralized operation.
 {When the centralized combining is considered, the centralized MMSE (C-MMSE) combining is given as \cite[Sec. 4.1]{bjornson2017massive}}
\begin{equation}\label{C-MMSE combining}
{\bf{v}}_k^{{\rm{C - MMSE}}} = {p_k}{\left( {\sum\limits_{i = 1}^K {{p_i}} {{{\bf{\hat h}}}_i}{\bf{\hat h}}_i^H + \sum\limits_{i = 1}^K {{p_i}} {{\bf{C}}_i} + \sigma _{{\rm{ul}}}^2{\bf{I}}} \right)^{ - 1}}{{{\bf{\hat h}}}_k}
\end{equation}
and minimizes the conditional MSE ${\mathbb{E}}\{ {{{| {{s_k} - {\bf{v}}_{k}^H{{\bf{y}}^{{\rm{ul}}}}} |}^2}|\{ {{{{\bf{\hat h}}}_{i}}} \}} \}$ between the desired signal and the centralized receive combined signal, where the expectation is computed conditioned on the centralized channel estimates. C-MMSE combining can also be shown to maximize the SINR of UE $k$ \cite{[162],[5]}.
{C-MMSE combining has a relatively high computational complexity but since the computation in \eqref{C-MMSE combining} is performed at the CPU, which is generally assumed to have high computational capability, it can be practically implementable. To reduce the complexity and enable distributed implementation, we will consider several MMSE-like combining methods that are specially developed for CF mMIMO.
As a distributed combining scheme directly inspired by C-MMSE}, AP $l$ can use local MMSE (L-MMSE) combining \cite{[162]}
\begin{equation}\label{L-MMSE_combining}
{\bf{v}}_{kl}^{{\rm{L - MMSE}}} = {p_k}{\left( {\sum\limits_{i = 1}^K {{p_i}} {{{\bf{\hat h}}}_{il}}{\bf{\hat h}}_{il}^H + \sum\limits_{i = 1}^K {{p_i}} {{\bf{C}}_{il}} + \sigma _{{\rm{ul}}}^2{\bf{I}}} \right)^{ - 1}}{{{\bf{\hat h}}}_{kl}}
\end{equation}
when detecting the data from UE $k$.
This combining scheme has received its name from the fact that it minimizes the conditional MSE ${\mathbb{E}}\{ {{{| {{s_k} - {\bf{v}}_{kl}^H{{\bf{y}}_l^{{\rm{ul}}}}} |}^2}|\{ {{{{\bf{\hat h}}}_{il}}} \}} \}$ between the desired signal ${{s_k}}$ and the local receive combined signal ${{\bf{v}}_{kl}^H{{\bf{y}}_l^{{\rm{ul}}}}}$ at AP $l$, where the expectation is computed conditioned on the local channel estimates.

It can be observed from (\ref{L-MMSE_combining}) and (\ref{C-MMSE combining}) that we need to compute all the $K$ MMSE channel estimates $\{ {{{{\bf{\hat h}}}_{il}}:i = 1, \ldots ,K} \}$ at any AP $l$ that is serving UE $k$.
Therefore, the total number of complex multiplications required by L-MMSE and C-MMSE combining schemes, unfortunately, grow with $K$, thus making the complexity unscalable.
To solve this issue, the alternative partial MMSE (P-MMSE) and local P-MMSE (LP-MMSE) are proposed in \cite{[163]}.
The main idea of P-MMSE and LP-MMSE is that the interference that affects UE $k$ is mainly generated by a small subset of the other UEs \cite{[5]}. Therefore, only the UEs that are served by partially the same APs as UE $k$ should be included in the expression in (\ref{L-MMSE_combining}) and (\ref{C-MMSE combining}). These UEs have indices in the set
\begin{equation}
{{\cal P}_k} = \{ {i:{\cal M}_k \cap {\cal M}_i \ne \emptyset} \}.
\end{equation}
By utilising ${{\cal P}_k}$, the P-MMSE and LP-MMSE combining vectors are given as
\begin{align}\label{PMMSE}\notag
{\bf{v}}_k^{{\rm{P - MMSE}}} &= {p_k}\left( \sum\limits_{i \in {{\cal P}_k}} {{p_i}} {{\bf{D}}_k}{{{\bf{\hat h}}}_i}{\bf{\hat h}}_i^{\rm H} {{\bf{D}}_k} \right.\\
&\left.+ {{\bf{D}}_k}\left( {\sum\limits_{i \in {{\cal P}_k}} {{p_i}} {{\bf{C}}_i} + {\sigma _{{\rm{ul}}}^2}{{\bf{I}}_{LN}}} \right){{\bf{D}}_k,} \right)^\dag {{\bf{D}}_k}{{{\bf{\hat h}}}_k},
\end{align}
\begin{equation}\label{PMMSE}
{\bf{v}}_k^{{\rm{LP - MMSE}}} = {p_k}{\left( {\sum\limits_{i \in {{\cal P}_k}} {{p_i}} \left( {{{{\bf{\hat h}}}_{il}}{\bf{\hat h}}_{il}^H + {{\bf{C}}_{il}}} \right) + {\sigma _{{\rm{ul}}}^2}{{\bf{I}}_{LN}}} \right)^{ - 1}}{{{\bf{\hat h}}}_{kl}}.
\end{equation}

The structure of MMSE-like combining is quite intuitive. Let us take L-MMSE as an example.
The matrix that is inverted in (\ref{L-MMSE_combining}) is the conditional correlation matrix ${{\bf{C}}_{{{\bf{y}}_l^{{\rm{ul}}}}}} = {\mathbb{E}}\{ {{{\bf{y}}_l^{{\rm{ul}}}}({{\bf{y}}_l^{{\rm{ul}}}})^{\rm H}|\{ {{{{\bf{\hat h}}}_{il}}} \}} \}$ of the received signal, given the current set of channel estimates.
The multiplication ${\bf{C}}_{{{\bf{y}}_l^{{\rm{ul}}}}}^{ - 1/2}{{\bf{y}}_l^{{\rm{ul}}}}$ corresponds to the whitening of the received signal; that is, ${\mathbb{E}}\{ {{\bf{C}}_{{{\bf{y}}_l^{{\rm{ul}}}}}^{ - 1/2}{{\bf{y}}_l^{{\rm{ul}}}}{{( {{\bf{C}}_{{{\bf{y}}_l^{{\rm{ul}}}}}^{ - 1/2}{{\bf{y}}_l^{{\rm{ul}}}}} )}^{\rm H}}|\{ {{{{\bf{\hat h}}}_{il}}} \}} \} = {\bf{I}}$.
If we denote the whitened combining vector as ${{\bf{u}}_{kl}}$, it is related to the original combining vector as ${{\bf{v}}_{kl}} = {\bf{C}}_{{{\bf{y}}_l^{{\rm{ul}}}}}^{ - 1/2}{{\bf{u}}_{kl}}$.
The highest desired signal power is now received from the spatial direction ${\bf{C}}_{{{\bf{y}}_l^{{\rm{ul}}}}}^{ - 1/2}{{{\bf{\hat h}}}_{kl}}$ and due to the whitening, which makes the total power equal in all directions, the interference plus noise power is lowest in this direction.
Hence, the optimal whitened combining vector can be selected as ${{\bf{u}}_{kl}} = {\bf{C}}_{{{\bf{y}}_l^{{\rm{ul}}}}}^{ - 1/2}{{{\bf{\hat h}}}_{kl}}$. This results into ${{\bf{v}}_{kl}} = {\bf{C}}_{{{\bf{y}}_l^{{\rm{ul}}}}}^{ - 1/2}{{\bf{u}}_{kl}} = {\bf{C}}_{{{\bf{y}}_l^{{\rm{ul}}}}}^{ - 1}{{{\bf{\hat h}}}_{kl}}$.
Therefore, MMSE-like combining is obtained by whitening followed by MR combining \cite{bjornson2017massive}.

Although the MMSE-like combining is optimal, the research literature contains other schemes as well. There are two main reasons for that. Firstly, the C-MMSE scheme has high complexity since there is an $LN \times LN$ matrix inverse in (\ref{C-MMSE combining}).
Secondly, the performance of MMSE-like schemes is hard to analyze mathematically, while there are alternative schemes that can give more insightful closed-form SE expressions.

\subsubsection{ZF-like Combining}

In centralized operation, if the channel conditions are good, we can neglect all the correlation matrices in (\ref{C-MMSE combining}) and obtain the regularized zero-forcing (RZF) combining vector \cite{[241],[105]}
\begin{equation}\label{RZF_combining}
{\bf{v}}_{k}^{{\rm{RZF}}} = {{{\bf{\hat H}}}}{\left( {{\bf{\hat H}}^{\rm H}{{{\bf{\hat H}}}} + \sigma_{\rm{ul}}^2{{\bf{P}}^{ - 1}}} \right)^{ - 1}}{{{\bf{\hat e}}}_k},
\end{equation}
{where ${\bf{\hat H}} = \left[ {{{{\bf{\hat H}}}_1^{\text{UE}}}, \cdots ,{{{\bf{\hat H}}}_K^{\text{UE}}}} \right] \in {\mathbb C}^{LN \times K}$, ${{\bf{\hat H}}_k^{\text{UE}}} = {\left[ {{\bf{\hat h}}_{1k}^T, \cdots {\bf{\hat h}}_{Lk}^T} \right]^T \in {\mathbb C}^{LN}}$} for $k \in \left\{ {1, \cdots ,K} \right\}$, ${\bf{P}} = {\rm{diag}}\left( {{p_1}, \ldots ,{p_K}} \right) \in {{\mathbb{C}}^{K \times K}}$, and ${{{\bf{\hat e}}}_k}$ is the $k$th column of ${{\bf{I}}_K}$. The RZF scheme can also be realized in a distributed fashion, which called local RZF (L-RZF) scheme
\begin{equation}
{\bf{v}}_{kl}^{{\rm{L-RZF}}} = {{{\bf{\hat H}}}_l}{( {{\bf{\hat H}}_l^{\rm H}{{{\bf{\hat H}}}_l} + \sigma _{{\rm{ul}}}^2{{\bf{P}}^{ - 1}}} )^{ - 1}}{{{\bf{\hat e}}}_k},
\end{equation}
where ${{{\bf{\hat H}}}_l} = [ {{{{\bf{\hat h}}}_{1l}}, \ldots ,{{{\bf{\hat h}}}_{Kl}}} ]$.
When the SNR is high, the combining expression in (\ref{RZF_combining}) can be further approximated as
\begin{equation}\label{ZF_combining}
{\bf{v}}_{k}^{{\rm{ZF}}} = {{{\bf{\hat H}}}}{( {{\bf{\hat H}}^{\rm H}{{{\bf{\hat H}}}}} )^{ - 1}}{{{\bf{\hat e}}}_k},
\end{equation}
under the name of centralized zero-forcing (ZF) combining.

Unlike the centralized ZF combining, full-pilot ZF (FZF) combining can suppress interference in a fully distributed, coordinated, and scalable fashion \cite{[41],[132]}. Besides, the computation of FZF combining has much lower complexity than centralized ZF. When considering mutually
orthogonal pilot sequences and uncorrelated Rayleigh channels between APs and UEs, the channel estimate of ${{\bf{h}}_{kl}}$ is given as
\begin{equation}
{{\bf{\hat h}}_{kl}} = {c_{kl}}\left( {\sum\limits_{k = 1}^K {\sqrt {p_k} } {{\bf{h}}_{kl}}{\pmb{\phi}} _{{t_k}}^H + {{\bf{N}}_l}} \right){{\pmb{\phi}} _{{t_k}}},
\end{equation}
where
\begin{equation}
{c_{kl}} \buildrel \Delta \over = \frac{{{\sqrt {p_k{\tau _p}} } {\beta _{kl}}}}{{{\tau _p}\sum\limits_{i \in {{\cal S}_k}} {p_t} {\beta _{kl}} + {\sigma^2}}},
\end{equation}
${{\pmb{\phi}} _{{t_k}}}$ is the ${\tau _p}$-length pilot signals assigned to UE $k$,
and ${{\bf{N}}_l} \in {{\mathbb{C}}^{N \times {\tau _p}}}$ is a Gaussian noise matrix with i.i.d. {${{\cal N}_{\mathbb{C}}}\left( {0,{\sigma^2}} \right)$} elements.
The local combining vector that AP $l$ selects for UE $k$, ${\bf{v}}_{{t_k}l}^{{\rm{FZF}}} \in {{\mathbb{C}}^{N \times 1}}$, is given by
\begin{equation}\label{V_FZF}
{\bf{v}}_{{t_k}l}^{{\rm{FZF}}} = {c_{kl}}{{{\bf{\bar H}}}_l}{( {{\bf{\bar H}}_l^{\rm H}{{{\bf{\bar H}}}_l}} )^{ - 1}}{{\bf{e}}_{{t_k}}},
\end{equation}
where
\begin{equation}
{{\bf{\bar H}}_l} = \sum\limits_{k = 1}^K {\sqrt {p_k} } {{\bf{h}}_{kl}}{\pmb{\phi}} _{{t_k}}^{\rm H}{\bm{\Phi}} + {{\bf{N}}_l}{\bm{\Phi}},
\end{equation}
which is connected to the respective channel estimate by
\begin{equation}
{{{\bf{\hat h}}}_{kl}} = {c_{kl}}{{{\bf{\bar H}}}_l}{{\bf{e}}_{{t_k}}},
\end{equation}
where $\pmb{\Phi}  = \left[ {{{\pmb{\phi}} _1}, \ldots ,{{\pmb{\phi}} _{{\tau _p}}}} \right] \in {{\mathbb{C}}^{{\tau _p} \times {\tau _p}}}$ and ${{\bf{e}}_{{t_k}}}$ denotes the ${t_k}$th column of ${{\bf{I}}_{{\tau _p}}}$.

Employing the partial FZF (P-FZF) combining leads to more array gain with the cost of only suppressing partial interference. Specifically, AP $l$ employs the P-FZF combining for strong UEs whose channel gain are large ${{\cal T}_l} \subset \left\{ {1, \ldots ,K} \right\}$ and MR combining for UEs with poor channel condition ${{\cal E}_l} \subset \left\{ {1, \ldots ,K} \right\}$. Since only strong UEs use the P-FZF combining, we define ${\tau _{{{\cal T}_l}}}$ as the number of different pilots used by the UEs ${ \in {{\cal T}_l}}$ and ${{\cal R}_{{{\cal T}_l}}} = ( {{r_{l,1}}, \ldots ,{r_{l,{\tau _{{{\cal T}_l}}}}}} )$ as the set of the corresponding pilot indices. Therefore, the pilot-book matrix for UEs ${ \in {{\cal T}_l}}$ is given by {${{\bm{\Phi}} _{{{\cal T}_l}}}{\rm{ = }}{\bm{\Phi }} {{\bf{E}}_{{{\cal T}_l}}}$,} where ${{\bf{E}}_{{{\cal T}_l}}} = ( {{{\bf{e}}_{{r_{l,1}}}}, \ldots ,{{\bf{e}}_{{r_{l,{\tau _{{{\cal T}_l}}}}}}}} ) \in {{\mathbb{C}}^{{\tau _p} \times {\tau _{{{\cal T}_l}}}}}$ and ${{{\bf{e}}_{{r_{l,i}}}}}$ is the ${{r_{l,i}}}$th column of ${{\bf{I}}_{{\tau _p}}}$. With respect to ${{\bm{\Phi}} _{{{\cal T}_l}}}$, we define ${j_{kl}} \in \left\{ {1, \ldots ,{\tau _{{{\cal T}_l}}}} \right\}$ the index. Let {${{\bm{\varepsilon}} _{{j_{kl}}}} \in {{\mathbb{C}}^{{\tau _{{{\cal T}_l}}} \times 1}}$} as the ${j_{kl}}$th column of ${{\bf{I}}_{{\tau _{{{\cal T}_l}}}}}$, and it leads to ${{\bf{E}}_{{{\cal T}_l}}}{\varepsilon _{{j_{kl}}}}{\rm{ = }}{{\bf{e}}_{{t_k}}}$. Then, the P-FZF combining for UE $k$ $\in {{\cal T}_l}$ at AP $l$ is given as
\begin{equation}\label{V_PFZF}
{{\bf{v}}_{{t_k}l}^{{\rm{P-FZF}}}} = {c_{kl}}{{{\bf{\bar H}}}_l}{{\bf{E}}_{{{\cal T}_l}}}{( {{\bf{E}}_{{{\cal T}_l}}^{\rm H}{\bf{\bar H}}_l^{\rm H}{{{\bf{\bar H}}}_l}{{\bf{E}}_{{{\cal T}_l}}}} )^{ - 1}}{\varepsilon _{{j_{kl}}}}.
\end{equation}

To improve the service quality of weak UEs, which is the main advantage of CF mMIMO compared with cellular systems, we can alternatively apply the protective weak P-FZF (PWP-FZF) combining for weak UEs to significantly reduce the intra-group interference. The main idea of PWP-FZF is to force the MR combining vector to take place in the orthogonal complement of ${{\bf{\bar H}}_l}{{\bf{E}}_{{{\cal T}_l}}}$, which is the effective channels of UEs in ${{\cal T}_l}$. With PWP-FZF, the MR combining used at AP $l$ for UEs in ${{\cal E}_l}$ is now given by
\begin{equation}
{{\bf{v}}_{kl}^{{\text{PMR}}} = {c_{kl}}{{\bf{J}}_l}{{{\bf{\bar H}}}_l}{{\bf{e}}_{{t_k}}}},
\end{equation}
where
\begin{equation}
{{\bf{J}}_l} = {{\bf{I}}_N} - {{{\bf{\bar H}}}_l}{{\bf{E}}_{{{\cal T}_l}}}{\left( {{\bf{E}}_{{{\cal T}_l}}^{\rm H}{\bf{\bar H}}_l^{\rm H}{{{\bf{\bar H}}}_l}{{\bf{E}}_{{{\cal T}_l}}}} \right)^{ - 1}}{\bf{E}}_{{{\cal T}_l}}^{\rm H}{\bf{\bar H}}_l^{\rm H}
\end{equation}
represents the projection matrix onto the orthogonal complement of ${{\bf{\bar H}}_l}{{\bf{E}}_{{{\cal T}_l}}}$.

\begin{figure}[t!]
\centering
\includegraphics[scale=0.64]{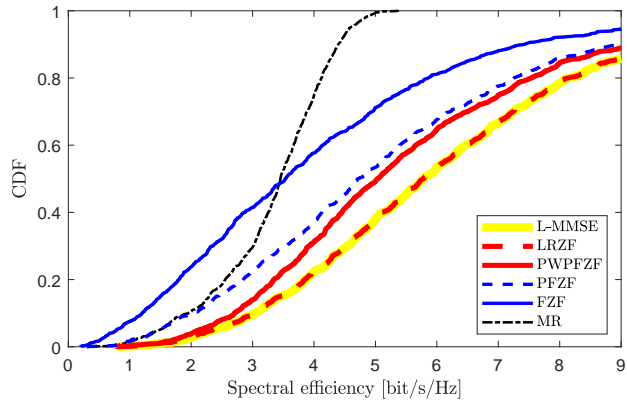}
\caption{CDF of the uplink SE per UE achieved by {MR \cite{bjornson2017massive}}, {L-MMSE \cite{[162]}}, and different {ZF-like \cite{[105]}} combining schemes. }
\label{fig:combining 1}
\end{figure}
\subsubsection{MMSE-SIC Combining}

All the combining schemes mentioned above are based on using linear receive combining. Still, another benefit of centralizing the signal processing at the CPU is that more advanced decoding methods can be used since system-wide CSI and high computational resources are available. The potential benefits of the MMSE-based non-linear successive interference cancelation (SIC) method are investigated in \cite{[162]}, which means that the CPU decodes one UE signal at a time, and then sequentially subtracts interference that the decoded signal caused to the remaining signals.
However, the numerical results in \cite{[162]} show that non-linear MMSE-SIC can only achieve a minor gain over the linear MMSE in terms of the average SE in uplink when the favorable propagation phenomenon exists. The reason is that the SIC method is only effective when there are a few strongly interfering UEs, while CF mMIMO is more characterized by having many UEs that cause little interference to each other.

In Fig.~\ref{fig:combining 1}, the cumulative distribution function (CDF) of the uplink SE per UE is shown for {the} MR, FZF, P-FZF, PWP-FZF, L-RZF, and L-MMSE combining schemes with $L=25$, $K=10$, $N=8$, ${\tau_p}=7$ and $p_k =100$ mW for each UE when using LSFD. The performance gap between the MR combining and the ZF-based schemes are quite significant, especially for UEs with large channel gains. It results from the impact of inter-user interference while FZF, P-FZF, and PWP-FZF combining schemes all can suppress that interference.
Besides, the advantage of employing P-FZF and PWP-FZF rather than FZF is noticeable. FZF spends ${\tau _p}$ degrees of freedom to cancel the pilot contamination and inter-user interference while P-FZF and PWP-FZF only spend ${\tau _{{{\cal T}_l}}}$ degrees of freedom and take advantage of a larger array gain. Compared with P-FZF, PWP-FZF gives a higher 95\%-likely SE, which is due to its protective nature of weak UEs with lower channel gain. Furthermore, Fig.~\ref{fig:combining 1} also shows that the performance gap between L-RZF and L-MMSE is quite small and they outperform the other schemes.

Figure~\ref{fig:combining 2} shows the CDF of the uplink SE per UE of scalable centralized P-MMSE and distributed LP-MMSE schemes with two benchmarks where all APs serve all UEs: C-MMSE combining in (\ref{C-MMSE combining}), distributed L-MMSE combining in (\ref{L-MMSE_combining}).
The first observation is that C-MMSE and P-MMSE outperform L-MMSE and LP-MMSE since the former two schemes exploit much more CSI than the latter two schemes to suppress the interference.
Then we can see that the scalable schemes (i.e., P-MMSE and LP-MMSE) provide almost the same performance compared to their counterparts (i.e., C-MMSE and L-MMSE).
The negligible performance loss comes from limiting the number of APs serving each UE, which is the price for scalability.
\begin{figure}[t!]
\centering
\includegraphics[scale=0.64]{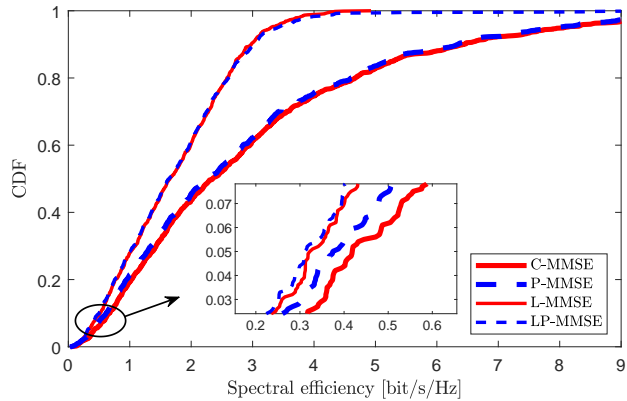}
\caption{CDF of the uplink SE per UE achieved by different {scalable \cite{[163]}} and {non-scalable combining \cite{bjornson2017massive}} schemes. }
\label{fig:combining 2}
\end{figure}

Next, we will first summarize the fronthaul costs, which refer to the amount of information to exchange via the fronthaul network to perform joint coherent transmission/detection and other centralized network operations of different distributed combining schemes.
More precisely, when using LSFD, for FZF, P-FZF, PWP-FZF, L-RZF and L-MMSE combining schemes, the number of complex scalars to send from the APs to the CPU via the fronthaul is $\left( {{\tau _c} - {\tau _p}} \right)KL$ in each coherence block or $KL + \left( {{L^2}{K^2} + KL} \right)/2$ for each realization of the UE locations/statistics.

Then, we summarize the computational complexity with MR, FZF, P-FZF, PWP-FZF, L-RZF, LP-MMSE and L-MMSE combining schemes per AP in terms of the number of complex multiplications in Table \ref{tab:complexity} according to \cite{[105],[241]}. Thanks to the fact that ${\tau _{{{\cal T}_l}}} \le {\tau _p}$, the complexity of P-FZF and PWP-FZF is lower than FZF and L-RZF. Compared with P-FZF, PWP-FZF needs $2\left( {{\tau _p} - {\tau _{{{\cal T}_l}}}} \right){\tau _{{{\cal T}_l}}}N$ more complex multiplications for computing the ${{\tau _p} - {\tau _{{{\cal T}_l}}}}$ MR combining vectors in (\ref{V_PFZF}).

\begin{table}[tp]
  \centering
  \fontsize{9}{12}\selectfont
  \caption{Computational complexity per AP in terms of number of complex multiplications}
  \label{tab:complexity}
    \begin{tabular}{ !{\vrule width1.2pt}  m{1.8cm}<{\centering} !{\vrule width1.2pt}  m{5.8cm}<{\centering}  !{\vrule width1.2pt} }

    \Xhline{1.2pt}
        \rowcolor{gray!50} \bf Scheme  & \bf Computational Complexity \cr
    \Xhline{1.2pt}

          MR & --   \cr\hline

        L-RZF \& FZF  & $\frac{3 N \tau^2_p}{2}+\frac{N \tau_p}{2}+\frac{\tau^3_p - \tau_p}{3}$ \cr\hline

        P-FZF  & $\frac{3 N \tau^2_{{\cal S}_l}}{2}+\frac{N \tau_{{\cal S}_l}}{2}+\frac{\tau^3_{{\cal S}_l} - \tau_{{\cal S}_l}}{3}$ \cr\hline

        PWP-FZF  & $\frac{3 N \tau^2_{{\cal S}_l}}{2}+\frac{N \tau_{{\cal S}_l}}{2}+\frac{\tau^3_{{\cal S}_l} - \tau_{{\cal S}_l}}{3}+2(\tau_p - \tau_{{\cal S}_l})N\tau_{{\cal S}_l}$ \cr\hline

        LP-MMSE  & $\frac{N^2 + N}{2}|{\cal D}_l|+\frac{N^3 - N}{3} +N^2$ \cr\hline

        L-MMSE &  $\frac{N^2 K + NK}{2}+\frac{N^3 - N}{3} +N^2$  \cr

    \Xhline{1.2pt}
    \end{tabular}

  \vspace{0cm}
\end{table}
\subsection{Transmit precoding}\label{subsec:combining precoding}

In the downlink, the acquired CSI is used to coherently precode the transmitted data signals, which is called \emph{transmit precoding}.
Transmit precoding means that each data signal is sent from multiple antennas, but with different amplitude and phases to direct the signal spatially \cite{[120]}. Each UE is affected by all the precoding vectors; the own precoding vector is multiplied with the channel response from the serving AP, while the other ones cause interference and are multiplied with the channel response from the corresponding transmitting APs. Hence, the precoding vectors should be selected carefully based on knowledge of the channel responses \cite{[121],[194],[106]}.
In this subsection, we comprehensively introduce the transmit precoding schemes used in CF mMIMO networks.

\subsubsection{Precoding via Uplink-Downlink Duality}

The precoding vector design is more complicated than that of the combining since the downlink SE of UE $k$ depends on the precoding vectors of all UEs in contrast to the uplink SEs that only depend on the UE's own combining vector ${{\bf{v}}_{k}}$.
The most commonly used approach for precoding design is employing uplink and downlink duality, which has been introduced in Lemma \ref{lemm:uplink downlink deality} in Section \ref{subsec:procedure}.
Based on that, all the linear combining schemes mentioned before can be utilized for designing the corresponding precoding vectors \cite{[32]}.

\subsubsection{Precoding via Utility Maximization}

The previous method only provides heuristic precoding vectors and then uses the downlink power allocation to further tune the performance of the UEs.
However, an optimal collection of centralized precoding vectors can be computed by maximizing a system-wide utility function \cite{[238],Bjornson2013d}. Suppose the utility is to maximize the minimum instantaneous SINR of all UEs. In that case, the optimal solution is obtained by solving a second-order cone program \cite{[238]} and arbitrary power constraints can be added to that problem \cite{Bjornson2013d}. Other metrics such as sum-rate maximization can be considered but only solved to local optimality. The benefit of optimal beamforming is that it outperforms any other scheme and comes with an optimal power allocation, while the drawback is the high computational complexity because each instance of the optimization problem is complex and needs to be solved once per coherence block. This approach cannot be made scalable.

\begin{figure}[t!]
\centering
\includegraphics[scale=1]{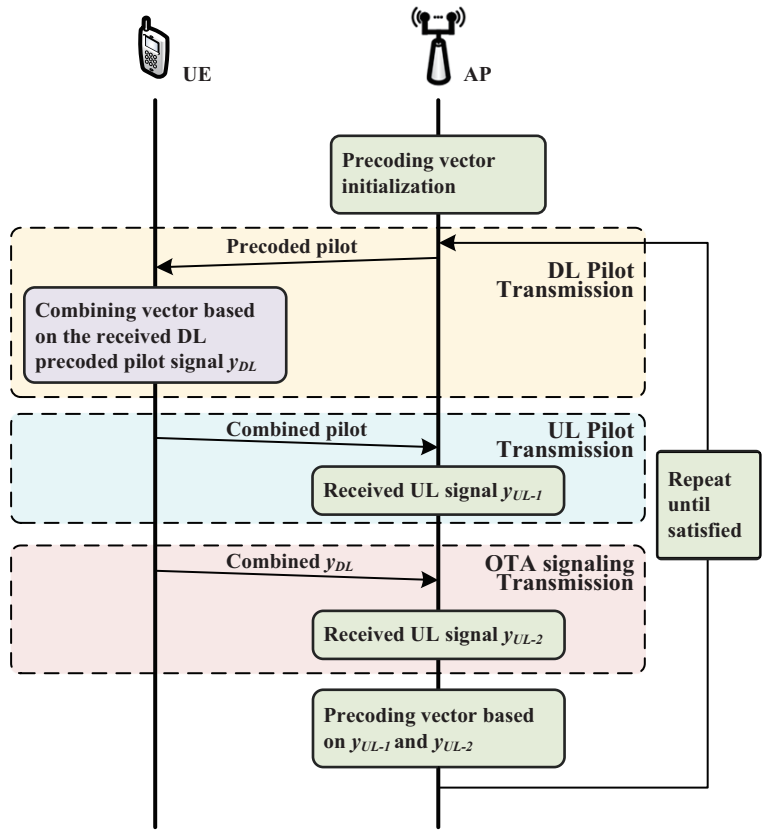}
\caption{Flow chart of the OTA-aided precoding design. }
\label{fig:ota}
\end{figure}

\subsubsection{Precoding via Over-the-air Signaling}
The precoding vectors should be optimized jointly for all UEs when not constructed based on the uplink-downlink duality. However, the distributed precoding design is hard to be carried out. The reason is that the optimization of ${{\bf{w}}_{kl}}$ needs information about the channel conditions between AP $k$ and the other APs and about the precoding vectors adopted by the latter for UE $k$. Such cross-term information is acquired utilizing fronthaul and must be adjusted iteratively, which leads to a sizeable fronthaul load. A novel over-the-air (OTA) approach which can optimize precoding vectors in a distributed fashion is proposed in \cite{[155],[180]}. It is based on a particular uplink signaling resource together with a new CSI combining mechanism. In this way, each AP can acquire the cross-term information over the air rather than via extensive fronthaul signaling.

The iterative implementation of the distributed precoding design via OTA signaling is illustrated in Fig.~\ref{fig:ota}.
To be specific, as a starting point, each AP initializes its precoding vector.
After that, each AP transmits a superposition of the pilots after precoding it with the corresponding precoding vector; each UE receives a downlink pilot signal ${\bf y}_{\rm DL}$ and computes its combining vector ${\bf v}_{\rm OTA}$ based on ${\bf y}_{\rm DL}$.
Then, each UE transmits its pilot combined after precoding it with its combining coefficient ${v}_{\rm OTA}$; each AP receives ${\bf y}_{\rm UL-1}$.
Besides, each UE transmits ${v}^{\rm H}_{\rm OTA}{\bf y}_{\rm DL}$ after precoding it with its combining coefficient ${v}_{\rm OTA}$; each AP receives ${\bf y}_{\rm UL-2}$.
Finally, each AP computes its precoding vector based on ${\bf y}_{\rm UL-1}$ and ${\bf y}_{\rm UL-2}$.
This procedure repeats until a predefined criterion is satisfied.
Unlike the conventional precoding design where the CSI among the APs is exchanged via fronthaul signaling, the aforementioned approach exchanges the CSI among the APs via the additional uplink OTA signaling resource (i.e., ${ v}_{\rm OTA}{ v}^{\rm H}_{\rm OTA}{\bf y}_{\rm DL}$), which advances in terms of scalability and flexibility.

\subsection{User Access and Association}\label{subsec:access}

When a UE is about to commence its communication, it first needs to access the network and then be assigned resources for the upcoming signal processing, such as pilot sequence, serving APs, etc.
In this subsection, we review the emerging schemes for the user access of the CF mMIMO communication network concerning AP selection, pilot assignment, user activity detection, and AP switch on/off strategies.

\subsubsection{AP Selection}

Compared to the cellular networks where each UE is only associated with one AP, CF systems require more fronthaul connections to transfer each UE's data to/from multiple APs, which leads to extra fronthaul provisioning and energy consumption.
However, to avoid substantial pilot contamination, each AP can only serve a limited number of UEs due to the pilot shortage.
For the above observations, the original design of CF mMIMO systems in \cite{[19]} wherein all UEs in the network are simultaneously served by all APs is unsuitable for the practical implementation of CF mMIMO.
Motivated by this, each UE should not be served by all APs, but a subset of selected APs, which is the so-called \emph{AP selection} and typically only serve at most UE per pilot.

There are two types AP selection schemes: large-scale-based scheme \cite{[11],[52]} and competition-based scheme \cite{[166]}.
In the former category, each UE $k$ selects $\left| {\cal M}_k \right| \le L$ dominant APs corresponding to the $\left| {\cal M}_k \right|$ largest large-scale fading coefficients, which satisfy
\begin{equation}\label{eq:}
  \sum\limits_{l \in {{\cal M}_k}} {\frac{{{{\bar \beta }_{kl}}}}{{\sum\nolimits_{j = 1}^L {{\beta _{kj}}} }}}  \ge \delta \% ,
\end{equation}
where $\left\{ {\bar \beta }_{kl} \right\}$ is the sorted (in descending order) version of the set $\left\{ { \beta }_{kl} \right\}$ and the threshold $\delta$ indicates the settled percentage of the total received power that signal APs contribute to each UE.
The large-scale-based selection scheme associates UEs and APs in a user-centric manner.
However, if an AP serves more than one UE per pilot, the signals from and to these pilot-sharing UEs will be strongly interfering, which is undesired.
Hence, an improved AP selection scheme could be designed from a UE perspective, but under constraints set by the APs' capabilities, e.g., an AP can only serve at most $\tau_p$ UEs.
For that purpose, a competition for an AP $l$ occurs when a new accessing UE $k$ attempts to select AP $l$ while AP $l$ already has $\tau_p$ UEs in ${\cal D}_l$.
The principle of the competition-based selection is that an AP $l$ gives priority to the UEs providing the best channel conditions.
Precisely, AP $l$ finds the ``weakest" UE
\begin{equation}\label{eq:}
k^* = \arg {\min _{i \in { \left\{{k}\right\} \cup {\cal D}_l }}}{\beta _{il}}.
\end{equation}
If $k^* = k$, UE $k$ puts AP $l$ into its \emph{blacklist}; otherwise, UE $k$ succeeds UE $k^*$ in ${\cal D}_l$, and UE $k^*$ puts AP $l$ into its blacklist likewise.
If a UE has $L-1$ APs on its blacklist, which means it has lost every competition it participated in. In that case, this UE will be associated with the only AP left and will not participate in another contest.
This operation prevents the weak UEs from being abandoned; at the end of the day, pilot-contamination issues can also be dealt with when designing the power allocation, so it is acceptable to occasionally let APs serve more than $\tau_p$ UEs and deal with the issue later.
The competition-based selection scheme allows the UEs to select as many serving APs as possible to make the best use of the APs service resources and meanwhile prevents an AP from serving more than $\tau_p$ UEs.
Both the described methods are heuristic, so there is room for improvements in the future.

\subsubsection{Pilot Assignment}

CSI is essential in multiple antenna systems, both cellular and CF.
It is usually acquired through pilot transmission between the UEs and APs.
However, the lack of a sufficient number of orthogonal pilot sequences, which comes from the natural channel variations in the time and frequency domain, compels the UEs to reuse the pilot resources, leading to pilot contamination.
This phenomenon reduces the channel estimation quality, making coherent transmission less effective and making it harder to reject interference between pilot-sharing UEs.
Thus, a properly designed pilot assignment algorithm is critical to ensure good performance in CF mMIMO systems.

Random assignment is a well-considered algorithm thanks to its simplicity \cite{[19],[103]}, wherein each UE is assigned a fixed pilot at random from the orthogonal pool and uses this pilot during the entire transmission.
This simple algorithm is not preferable since the neighboring UEs will occasionally use the same pilot and thus create strong mutual interference that is hard to suppress.
A step forward is the greedy algorithm which iteratively updates the pilot of the worst-performing UE after a random assignment \cite{[19],[55],[72]}. Greedy algorithms can converge to local optima but are unlikely to provide a globally optimal pilot assignment.
Another step is user-centric clustering \cite{[61]}, where the UEs are clustered into groups (joint or disjoint) based on the large-scale information, like large-scale coefficients, location of UEs and APs, and distance between UEs and APs.
The pilots are reused in/over these groups.
Several algorithms are considered in the clustering pilot assignment \cite{[191],[182],[192],[165]}.
Graph coloring is used in \cite{[191],[182]}, where the interference between UEs is modeled as a graph.
UEs (shown as vertices in the graph) are connected if at least an AP serves them.
Conventional graph coloring algorithms can be exploited to color the UEs with the fewest colors.
The final assignment is achieved by updating the interference graph.
To avoid being trapped in a local optimum, tabu search is used in \cite{[192]}, where the \emph{tabu list} records previous assignments to ensure the efficient search of the assignment solution space.
The authors of \cite{[165]} propose an iterative approach based on the Hungarian Algorithm.
In each iteration, each UE and its neighboring UEs are assigned with mutual orthogonal pilots by exploiting the Hungarian algorithm, given the pilot assignment of the rest of the UEs is fixed.
The final assignment is achieved when the performance measures reach convergence, or the iterations reach the allowed maximum number.

Although the aforementioned pilot assignment algorithms limit the pilot contamination to varying degrees, they might not be feasible for practical implementation since the complexity grows polynomially with the number of APs and UEs.
Therefore, provably scalable algorithms for pilot assignment have also been developed.

A joint AP selection and pilot assignment algorithm are proposed in \cite{[163]}, in which a UE first appoints its \emph{Master AP}.
Then the Master AP assigns the pilot with the least interference to this UE and informs a limited set of neighboring APs that it is about to serve this UE on the assigned pilot.
A neighboring AP decides to serve this UE on the assigned pilot or not based on its serving status.
This algorithm achieves scalability by providing each UE with the least bad pilot but performs no optimization for which pilot assignment is fairly straightforward.
The K-means algorithm is used for user clustering in \cite{[25],[166]}, where the UEs
are separated into disjoint clusters based on the knowledge of the location and interference relationship of UEs and APs, respectively.
The UEs in the same cluster are assigned with mutual orthogonal pilots.
Although the K-means algorithm separates the UE clusters as far as possible, it operates on the cluster level; or in other words, it dynamically divides the network into subareas, but it cannot prevent the neighboring UEs in different subareas from sharing the same pilot.
To solve this issue, the authors of \cite{[166]} propose a user-group algorithm operating in the UE level, where the interference relationship between the UEs is exploited to separate the UEs into disjoint groups iteratively.
The UEs in the same cluster share the same pilot.

\begin{figure}[t]
\centering
\includegraphics[scale=0.64]{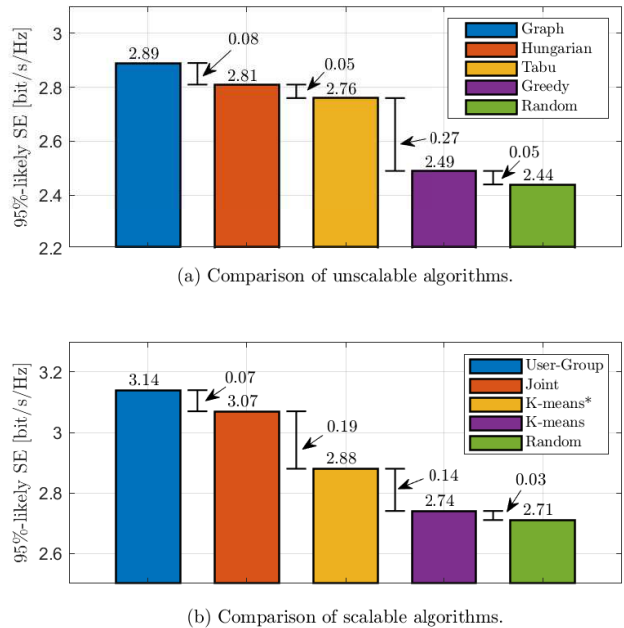}
\caption{95\%-likely SE with different pilot assignment algorithms. In abbreviation, the unscalable algorithms: Greedy \cite{[19]}, Tabu \cite{[192]}, Hungarian \cite{[165]}, and Graph \cite{[191]}. The scalable algorithms: Random \cite{[19]}, K-means \cite{[25]}, K-means* \cite{[166]}, Joint \cite{[163]}, and User-Group \cite{[166]}.
\label{fig:pilot assignment}}
\end{figure}

\begin{table*}[t]
  \centering
  \fontsize{9}{12}\selectfont
  \caption{Concise Description of Pilot Assignment Algorithms.}
  \label{tab:pilot assignment}
    \begin{tabular}{ !{\vrule width1.2pt}  m{2.3cm}<{\centering} !{\vrule width1.2pt}  m{7.8cm}<{\centering} !{\vrule width1.2pt}  m{4.3 cm}<{\centering} !{\vrule width1.2pt} }

    \Xhline{1.2pt}
        \rowcolor{gray!50} \bf Scheme  & \bf Concise Description  & \bf Online Complexity   \cr
    \Xhline{1.2pt}

         Random assignment \cite{[19]} & Each UEs is assigned a fixed pilot from the orthogonal. & ${\cal O}\left( K \right)$  \cr\hline

        Greedy \cite{[19],[55],[72]} & Start from a random assignment, iteratively update the pilot of the worst performing UE. & ${\cal O}\left( {KL} \right)$ \cite{[19]}  \cr\hline

        Dynamic Pilot Reuse \cite{[61]}  & Based on user-centric concept, two UEs with a large distance can share the same pilot. & ${\cal O}\left( {{{\mathbb N}{\mathbb T} } } K^2 \right)$ \cr\hline

        Graph Coloring \cite{[182],[191]}  & Interference is modeled as a graph and graph coloring algorithm is exploited to assign all UEs with fewest pilots. & ${\cal O}\left( {K^2 + KL + KL{{\log }_2}L} \right) $ \cite{[191]} \cr\hline

        Tabu-Search \cite{[192]}  & \emph{Tabu list} prevent the assignment from being trapped in the local optimum and ensure the efficient search of the assignment solution space. & ${\cal O}\left( {{N_{{\rm{tabu}}}}{K^2}L} \right)$ \cr\hline

        Hungarian Algorithm \cite{[165]}  & Each UE and its neighboring UEs are assigned with mutual orthogonal pilots by exploiting Hungarian Algorithm, given the pilot assignment of the rest UEs is fixed. & ${\cal O}\left( {K\left( {L + \tau _p^2} \right)} \right)$  \cr\hline

        Joint AP Selection and Pilot Assignment \cite{[163]}  & Each UE pinots its \emph{Master AP}, which assigns the pilot with least interference to this UE and informs the neighboring APs to cooperatively serve this UE with the assigned pilot. & ${\cal O}\left( {L + K} \right)$ \cr\hline

        K-Means \cite{[25],[166]} & K-means algorithm is used to separate the UE into disjoint clusters. The UEs in the same cluster are assigned with mutual orthogonal pilots. & ${\cal O}\left( {K^2/\tau_p + \tau_p^2 {\left\lceil K / \tau_p - 1 \right\rceil }} \right)$ \cite{[166]}  \cr\hline

        User-Group \cite{[166]} & Interference relationship between the UEs are exploited to iteratively separate the UEs into disjoint groups. The UEs in the same cluster share the same pilot. & ${\cal O}\left( { {K^2}L } \right)$  \cr

    \Xhline{1.2pt}
    \end{tabular}

  \vspace{0cm}
\end{table*}

A concise description of the pilot assignment algorithms is shown in Table \ref{tab:pilot assignment}, where the online complexity analysis is given.
Moreover, a comparison of the pilot assignment algorithms on the 95\%-likely SE is illustrated in Fig.~\ref{fig:pilot assignment} with the simulation setup of $L=100$, $K=50$, $\tau_c = 200$, $\tau_p =10$, and area of $0.5\times0.5$ km$^2$.
Although some algorithms outperform others in this given setup, the situation might change in another simulation setup.
But anyway, the algorithms with more complicated processing mechanisms and employing more information of the APs, UEs, and communication environment conditions will offer better performance.
Nevertheless, almost all these above algorithms are either unscalable or heuristic, which encourages us to look into the new scalable alternatives combined with the user-centric clustering approach for the optimized pilot assignment in future works.
Machine learning might offer a potential solution thanks to its powerful signal processing ability to deal with the highly loaded networks and the relatively low complexity of online computing and framework.

\subsubsection{User Activity Detection}

User activity detection is considered in highly crowded scenarios, like the Internet of Things (IoT) and the Internet of Everything (IoE).
One feature of these scenarios is that the devices only send small amounts of data and should be energy efficient. Thus the overhead due to random access and scheduling is extensive compared to the data;
the other is the sporadic nature of the transmission, i.e., only a relatively small fraction of UEs stays active with short-length payloads.
Both features make the grant-free access scheme a promising solution where the active UEs transmit their pilots and payloads simultaneously without scheduling in advance.
As the networks continue to densify, the user activity detection of the grant-free access is about to be a non-negligible problem.

Authors in \cite{[178]} formulate user activity detection as a maximum likelihood problem.
Based on coordinate descent, a detection algorithm with affordable complexity is provided.
Authors in \cite{[189]}, on the other hand, formulate the user activity detection as compressive sensing (CS) problem by exploiting the angular-domain sparsity of the CF mMIMO channels, where the OFDM technique is used for uplink transmission.
Both of the above algorithms employ non-orthogonal pilot sequences for user identification.

Achieving massive user signatures by sacrificing the orthogonality of the pilot sequences will reduce the channel estimation quality, which deteriorates the system SE performance.
How to efficiently juggle massive access and high spectral efficiency is an open issue in the considered highly crowded CF mMIMO systems.

\subsubsection{AP Switch On/Off}

With the growing demand for green communications, AP switch On/Off (ASO) strategy design is becoming a rasing topic in CF mMIMO networks. Some APs are dynamically turned On/Off based on the location and data traffic generated by the served UEs.
The rationale behind ASO in CF mMIMO is that a large number of APs are implemented in the network, and their neighboring APs could likely fill the SE requirement of the UEs.
The goal of ASO is efficiently exploiting some, not all, competent APs for serving the dynamic traffic load requests to improve the system energy efficiency (EE), which will be defined later, and decrease the carbon footprint.

Unlike the other works in CF mMIMO considering a static network where the APs are always active, and status of them is irrelative with when and where the traffic load requests come from, \cite{[168],[174],[179]} treat the status of the APs as an optimization variable to serve the UEs in a more efficiently fashion.
A globally optimal solution of ASO is provided in \cite{[168]} by solving a mixed-integer second-order cone program (SOCP).
Due to the high computing complexity of the non-convex optimization problem, two heuristic low-complexity algorithms are also proposed by utilizing the channel sparsity structure.
Authors in \cite{[174],[179]} develop a collection of heuristic AP switch On/Off algorithms based on the location and propagation losses between APs and UEs, namely Random selection ASO (RS-ASO), Chi-square test-based ASO (ChiS-ASO), Kolmogorov-Smirnov test-based ASO (KS-ASO), Logarithmic statistical energy ASO (LSE-ASO), Minimum propagation losses-aware ASO (MPL-ASO), Optimal EE-based greedy ASO (OG-ASO), etc.
Among them, the algorithms based on goodness-of-fit techniques (i.e., ChiS-ASO, KS-ASO, and LSE-ASO) significantly outperform RS-ASO by trying to match the spatial distribution of active APs to one of the UEs.
OG-ASO offers the best SE performance by wielding the knowledge of spatial correlation matrices, power control matrices, power consumption metrics.
However, MPL-ASO makes a good tradeoff between SE performance and complexity by exploiting the large-scale fading coefficients between the APs and UEs with a minor performance penalty.

\subsection{Power Control and Power Allocation}

To obtain good system performance, the available radio resources must be efficiently managed.
To be specific, the $K$ UEs must select  appropriate transmit powers $0 \le p_k \le p_{\rm max}$, $k = 1,\ldots,K$ during uplink transmission, while the $L$ APs must allocate their transmit power $0 \le \rho_l \le \rho_{\rm max}$, $l = 1,\ldots,L$ during downlink transmission.
The procedure of controlling the uplink transmit powers is called \emph{power control}, while the procedure of allocating the downlink transmit power between UEs is called \emph{power allocation}.
Ideally, the power control/allocation should be carried out to optimize some system-wide utility function, describing the system performance as a whole.
Since the utility function determines the structure of the optimization problem and thereby which approaches can be applied to solve it, we survey the state-of-the-art power control/allocation algorithms for solving the three most common types of utility optimization problems: max-min fairness, max sum SE, and max EE.

\subsubsection{Max-Min Fairness}\label{subsubsec:maxmin}

The importance of the max-min fairness utility was emphasized in the early works \cite{[18],[19]}, where the vision of CF mMIMO was to provide uniformly good service over the entire coverage area.
The goal is to maximize the lowest SE among all the UEs in the network, which leads to uniform service, while the channel conditions will determine how good that service quality is.
Since the SE of UE $k$ is an increasing function of the effective SINR (see \eqref{eq:SE ul 2} and \eqref{eq:SE dl}), maximizing the lowest SE is equivalent to maximizing the lowest effective SINR among all the UEs.

There are several instances of the max-min fairness problem that can be shown to be convex or quasi-convex. Thus the optimal solution can be obtained by exploiting bisection search, and convex optimization \cite{[140]}, geometric programming (GP), or SOCP \cite{[18],[19],[238]}. The SOCP formulation dates back to \cite{Bengtsson2001a}.
The authors in \cite{[140]} formulated a weighted max-min power optimization problem in a multigroup multicast CF mMIMO system that could be cast as a quasi-concave problem via a quadratic convex transformation.
Moreover, \cite{[238]} developed an optimum downlink beamforming method in CF mMIMO systems by solving a max-min problem that maximizes the minimum SINR among all UEs.

There are instances of the max-min fairness problem that are non-convex, in which case one can sometimes find a local optimum by \emph{alternating optimization}, which partitions the optimization variables into several sets and cyclically optimizes one at a time while keeping the
other variable sets fixed \cite{[12],[172]}.
In this way, the original problem can be effectively decomposed into several subproblems, which could be convex and then be cyclically solved by exploiting the approaches applied for the convex cases (e.g., bisection search, GP \cite{[82],[28],[83],[26]}, or SOCP).
For example, \cite{[28]} considered a mixed quality-of-service (QoS) problem, where the minimum SE of non-real-time UEs is maximized while the rates of the real-time UEs meet their target rates.
The original non-convex problem was decomposed into two sub-problems wherein the GP was exploited to solve the power allocation problem.
The same approach was also applied in \cite{[82],[83],[26]} for solving the max-min fairness problem with power constraints.


Although the aforementioned algorithms optimize the transmit
powers for all UEs to maximize the lowest SE in a system-wide manner, it is unavoidable that their computational complexities grow unboundedly with the number of the UEs, $K$, which makes these algorithms unscalable according to the definition in Section \ref{subsec:scalable}.
Hence, distributed and heuristic schemes are needed to obtain the scalable power control in large, practically implementable CF networks. Each device makes a local decision with limited involvement of the other devices.

Fractional power control is a classical heuristic scheme in uplink multiuser systems.
The principle of fractional power control is controlling the UE transmit power to compensate for a fraction of the pathloss differences among the UEs that are partially served by the same APs, where UE $k$ selects its uplink to transmit power as \cite{[119],nikbakht2020uplink,[166],cellfreebook}
\begin{equation}\label{eq:fractional ul}
{p_k} = \frac{{{{\min }_{i\in {\cal S}_k}}{{\left( {\sum\nolimits_{l \in {{\cal M}_i}} {{\beta _{il}}} } \right)}^\nu }}}{{{{\left( {\sum\nolimits_{l \in {{\cal M}_k}} {{\beta _{kl}}} } \right)}^\nu }}}{p _{\rm max}}
\end{equation}
where the exponent $\nu \in [0,1]$ dictates the power control behavior.
The nominator in \eqref{eq:fractional ul} forces $p_k \in [0,p_{\rm max}]$.
Note that ${ {\sum\nolimits_{l \in {{\cal M}_k}} {{\beta _{kl}}} }  }$ denotes the total channel gain from UE $k$ to the APs that serve it, of which value is large when UE $k$ is in good channel condition.
A larger value of $\nu$ encourages each UE to compensate for the variations in the total channel gain among the
UEs in $i \in {\cal S}_k$, which promotes more fairness.
If $\nu = 0$, then all UEs transmit with maximum power (i.e., $p_k =p_{\rm max}$  ), which is the so-called \emph{equal power allocation} or \emph{full power transmission}.

Fraction power allocation can be used in the downlink, in which case AP $l$ selects the downlink power allocation coefficient for UE $k$ proportionally to the channel gain, $\beta_{kl}$, as \cite{[163],cellfreebook}
\begin{equation}\label{eq:fractional dl}
{\rho _{kl}} = {\begin{cases}
{\frac{{ ({{\beta _{kl}}})^\upsilon }}{{\sum\nolimits_{i \in {{\cal D}_l}} { ({{\beta _{il}}})^\upsilon } }}{\rho _{\rm max}}}&{{\rm if\ } k\in{\cal D}_l}\\
{0}&{\rm otherwise},
\end{cases}}
\end{equation}
where the exponent $\upsilon \in [0,1]$ dictates the power control behavior.
A larger value of $\upsilon$ gives higher emphasis to the UEs according to their respective channel gains while $\upsilon = 0$ indicates that each UE in ${\cal D}_l$ is allocated with equal power ${\rho _{kl}} = \frac{{\rho _{\rm max}}}{|{\cal D}_l|}$.
This leads to allocating more power to the UEs in better channel conditions, which seems to be contrary to the max-min fairness concept. Still, it is generally not since the UEs in good channel conditions will then avoid high interference.
Since the essence of fraction power control/allocation is to impose a structure with a parameter ($\nu$ or $\upsilon$) that can be tuned, we can also apply it to other utility functions and adjust the parameter to identify a suitable heuristic solution.

Apart from traditional optimization and heuristic methods, ML can be utilized to design power control/allocation methods.
Such an approach cannot provide a better solution than the one found by classical optimization methods, but it could potentially lower the computational performance \cite{bjornson2020two}.
For instance, one can greatly reduce the online computational complexity at the price of offline training \cite{[236]}.
There are ML-based schemes proposed to solve the max-min fairness problem in CF mMIMO \cite{[91],[236],[ML3]}.
Moreover, \cite{[91]} proposed to approximately solve the max-min fairness problem using local CSI by training a neural network to identify a mapping between that local CSI and the optimal solution to the system-wide max-min fairness problem.
An unsupervised learning approach was studied in \cite{[ML3]}, which took the large-scale fading coefficients as inputs to learning the map between these coefficients and soft max-min and max-prod power control policies.

\subsubsection{Max Sum SE}

A potential drawback of the max-min SE fairness problem is that a few UEs might drag down the overall system performance with bad channel conditions.
The overwhelming majority of UEs in a large network can likely achieve substantially larger SEs while barely affecting the UEs in the worst conditions since every UE only causes interference to a small subset of neighboring UEs.
This motivates the maximization of the sum SE, which represents the overall SE performance of the network instead of the SE achieved by a specific UE.

The max sum SE problem is usually not convex (see \eqref{eq:SE ul 2} and \eqref{eq:SE dl}); hence, it is hard to obtain the optimal solution, and we typically need to settle for a local optimum.
The aforementioned alternating optimization can be applied for addressing the non-convexity. The classical weighted MMSE method can be utilized to optimize the data powers \cite{cellfreebook}, when the pilot powers are fixed.
The problem of joint data and pilot transmit power control was considered in \cite{[66]} using the Lagrange multiplier method.
Another effective approach is exploiting the \emph{successive convex approximation (SCA)}, which is an iterative algorithm to maximize sum SE by employing convex optimization where the non-convex term is substituted by its convex approximation \cite{[54],[96],[229],[70],[144]}.
A max sum SE optimization problem for a downlink setup with hardware impairment was considered in \cite{[229]}.
The problem was non-convex and, thus, the SCA policy was employed by reformulating the original problem as a SOCP.
The same approach was also applied in \cite{[70]} and \cite{[144]} when short-term power constraints and low-resolution analog-to-digital converters (ADCs) were considered, respectively.

Since the aforementioned methods are searching for a local optimum, ML-based schemes can potentially find better solutions \cite{[76],[ML4],[ML5]}.
{A max sum SE problem in an uplink CF mMIMO system} was studied in \cite{[76]} using artificial neural networks (ANNs), in which the UE positions were taken as input and the power control policy as output.
In \cite{[ML4]}, a deep convolutional neural network (DCNN) was considered in an uplink CF mMIMO system with limited-fronthaul, where the LSF information was exploited to predict the max sum SE power control policy.
A deep neural network-based power control method was proposed in  \cite{[ML5]}. 

\begin{table*}[tp]
  \centering
  \fontsize{9}{12}\selectfont
  \caption{Power Control/Allocation Approaches.}
  \label{tab:power optimization}
    \begin{tabular}{ !{\vrule width1.2pt}  m{2.5cm}<{\centering} !{\vrule width1.2pt}  m{4 cm}<{\centering} !{\vrule width1.2pt} m{4cm}<{\centering} !{\vrule width1.2pt}  m{3.5 cm}<{\centering} !{\vrule width1.2pt}}
    \Xhline{1.2pt}
        \rowcolor{gray!50} \bf Utility Function & \bf Max-min Fairness & \bf Max sum SE & \bf Max EE  \\
    \Xhline{1.2pt}
    Alternative Optimization & \cite{[82],[12],[172],[28],[83],[26]}     & \cite{[66]}     & \cite{[74]} \\\hline
    SCA   & --     & \cite{[54],[96],[229],[70],[144]}     & \cite{[21],[86]} \\\hline
    Bisection & \cite{[140]}     & --     & -- \\\hline
    GP    & \cite{[82],[28],[83],[26]}     & --     & \cite{[86]} \\\hline
    SOCP  & \cite{[238]}     & \cite{[229],[70],[144]}     & \cite{[21]} \\\hline
    Fractional & \multicolumn{3}{c|}{{\cite{[163],[119],[166],nikbakht2020uplink}}}  \\\hline
    Others & --     & \cite{[66]}     & \cite{[128]} \\\hline
    ML-based & \cite{[236],[91],[ML3]}     & \cite{[76],[ML4],[ML5]}     & -- \\\hline

    \Xhline{1.2pt}
    \end{tabular}
  \vspace{0cm}
\end{table*}

\subsubsection{Max EE}

When designing a large CF mMIMO network, the energy efficiency (EE) is another essential performance metric to consider, which indicates not only ``how much and fast" the information can be transmitted, but also ``how economically" in terms of the energy, i.e., how much energy it takes to reliably transmit a certain amount of information \cite{[21],[53],[69]}.
Technically, the EE is defined as \cite{bjornson2017massive}
\begin{equation}\label{eq:ee}
 {\sf{EE}} = \frac{B \cdot \sum\nolimits_{k = 1}^K {{\sf{SE}}_k}}{{P_{\rm total}}}
\end{equation}
where $B$ is the system bandwidth and ${P_{\rm total}}$ is the total power consumption which usually includes four main terms: the transmit powers, a term accounting for the analog processing the transceiver chains, a term accounting for the digital signal processing, and a term for the fronthaul connections.

The max EE problem in \eqref{eq:ee} is non-convex, thus the alternating optimization \cite{[74]} and SCA \cite{[21],[86]} methods are normally used to solve the problem.
The EE maximization problem in a mmWave CF mMIMO system was considered in \cite{[74]}.
Since it is a non-convex problem, the successive power-bound maximization method, which is based on the idea of merging the alternating optimization and sequential convex programming, is used for alternatively optimizing the transmit powers of each AP while keeping the transmit powers of the other APs fixed.
Moreover, \cite{[86]} studied the EE maximization problem with quantization, considering per-UE power, fronthaul capacity, and throughput requirement constraints.
To solve the non-convex problem, the original SCA problem was decoupled into two sub-problems, namely, receiver filter coefficient design and power control.
The former was formulated as a generalized eigenvalue problem, while the GP addressed the latter after exploiting an SCA and a heuristic sub-optimal scheme.
The SCA method was also exploited in \cite{[21]} to maximize the total EE under the per-AP and per-UE power constraints.
The problem was approximately solved via a sequence of SOCP.
Although the second-order optimization methods have performed very well, their complexities do not scale favorably with the network size, which motivates first-order methods.
Finally, \cite{[128]} proposed a first-order method for non-convex programming to the EE maximization problem.
This first-order method could achieve the same performance with a faster run time than the second-order methods.

In Table \ref{tab:power optimization}, we summarize the optimization utilities, power optimization methods, and the main literature.

\section{Practical implementation}\label{sec:practical}

Although many sophisticated schemes and algorithms have been designed for the CF mMIMO systems, they are mainly based on simplifying assumptions such as error-free fronthaul connections and perfect hardware, which are unlikely to hold in practical deployment.
Hence, in this section, we will discuss the practical issues that a CF mMIMO system has to face to implement (i.e., fronthaul limitation and hardware impairment) and survey the existing schemes to address them.
Although synchronization is also an essential issue for practical implementation in CF mMIMO systems, little work has been done on this topic. Hence, we will discuss this part in the Future Research Directions.

\subsection{Fronthaul} \label{subsec:fronthaul}
One of the main issues for CF mMIMO systems is the limited capacity of the fronthaul links from the APs to the CPU \cite{interdonato2019ubiquitous,[88],[164]}. {Due to a large number of antennas at the APs, a large number of signals should be exchanged between APs and the CPU through the fronthaul links and hence cause huge power consumption (e.g., 0.25 W/(Gbits/s) when using optical fiber cables \cite{[53]}).} Besides, when converted to digital form, it requires a huge capacity for the fronthaul links many times the corresponding user data rate in the uplink to ensure signals are transferred with sufficient precision \cite{bashar2020exploiting}. In the C-RAN literature, this has been estimated as 20-50 times the corresponding data rate, implemented using the common public radio interface (CPRI) standard \cite{ericsson2015common}, typically over optical fiber. Therefore, reducing the fronthaul load constitutes one of the most substantial challenges in practical CF mMIMO systems \cite{interdonato2019ubiquitous,bashar2019uplink,burr2018cooperative,bashar2018performance}.

There are two approaches to address the capacity-limited fronthaul issue: (1) quantizing the transmit signals; and (2) using structured lattice codes.

\subsubsection{Quantizing}

Using a small number of bits to quantize the transmit signals is feasible to reduce the fronthaul load. Therefore, employing low-resolution ADCs at APs is a promising and practical solution, facilitating low power consumption and small hardware cost. Depending on how the APs process and forward the signals to the CPU, there are four main types of transmission in the uplink:

{\bf Compress-forward-estimate (CFE) \cite{masoumi2019performance,masoumi2019transmission,bashar2020exploiting,bashar2020deep,bashar2018performance,femenias2020fronthaul,bashar2018cell}} Each AP compresses the received pilot and data signals separately and forwards the compressed versions over the fronthaul link to the CPU. Then, the channel estimation, the design of combining vectors, and the data recovery are carried out at the CPU. This is a way to implement centralized combining methods over a limited fronthaul network by compression at the APs.

{\bf Estimate-compress-forward (ECF) \cite{masoumi2019performance,masoumi2019transmission}} First, the channel estimation is performed at each AP. Each AP separately compresses the estimated channels and data signals and forwards the former to the CPU. Finally, the CPU recovers the CSI and performs data detection using centralized combining. Since the compression is implemented at the APs, ECF reduces the fronthaul load in a distributed fashion.

{\bf Estimate-multiply-compress-forward (EMCF) \cite{bashar2019energy}} Each AP first estimates the channels, then multiplies the received data signal by the local combining vector (computed based on the local channel estimate) and compresses and forwards the results to the CPU. Thus, CPU only performs data detection. Therefore, when using EMCF, the design of combining vectors and the compression is implemented at the APs in a distributed way.

{\bf Estimate--multiply-compress-forward-weight (EMCFW) \cite{bashar2019energy,femenias2020fronthaul,bashar2018cell}} Similar to EMCF, the signal is further multiplied by receiver filter coefficients at the CPU to improve the performance. The design of combining vectors and the compression is also implemented in a distributed way.

When it comes to the downlink, two approaches have been considered:

{\bf Compress-after-precoding (CAP) \cite{femenias2019reduced,femenias2019cell,parida2018downlink}} The centralized signal with precoding is first computed and then compressed at the CPU before being sent to the APs, which makes CAP suitable for the centralized precoding.

{\bf Precoding-after-compress (PAC) \cite{boroujerdi2019cell}} A simple compression is done at the CPU where the symbol for each of the UEs is separately quantized. Then, each AP receives the symbols and designs the precoding vectors for each UE, which makes PAC suitable for the distributed precoding.

\begin{figure}[t]
\centering
\includegraphics[scale=0.6]{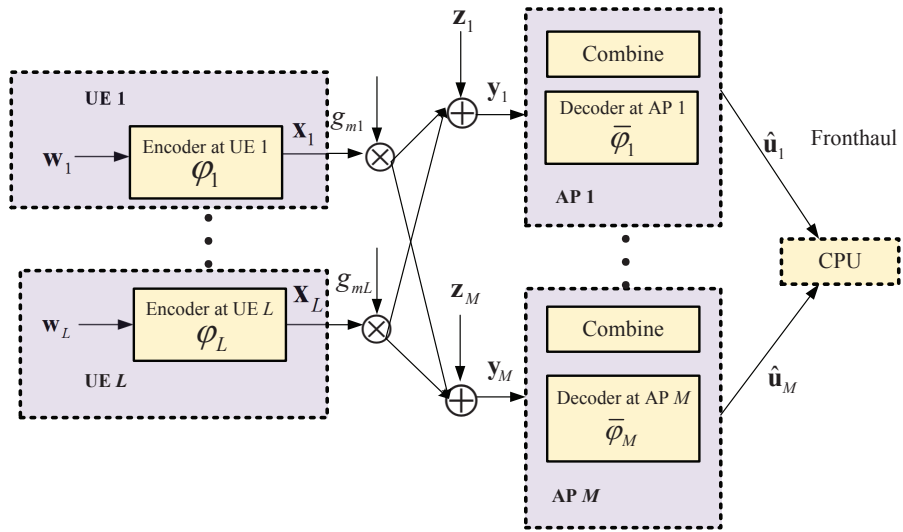}
\caption{A C\&F framework for CF mMIMO systems.}
\label{fig:ecf}
\end{figure}

\subsubsection{Compute-and-Forward}

Another possible solution could be the compute-and-forward (C\&F) schemes \cite{huang2017compute,zhang2019expanded}, which can reduce the fronthaul load efficiently by decreasing the cardinality of symbols transmitted to the CPU. As shown in Fig.~\ref{fig:ecf}, in the C\&F scheme, each AP forwards an integer-linear combination of the transmitted signals of all UEs with the same cardinality as each UE's signal. The served UEs can be determined by selecting the coefficient vector and are limited by the computation rate. Although the C\&F scheme can obtain the theoretical minimum fronthaul requirement of CF mMIMO with achieving lossless transmission, it is only suitable for the symmetric scenario where all UEs transmit with equal power. However, in CF mMIMO, the UEs can be allocated with unequal transmit power to compensate for pathloss differences. Hence, the expanded compute-and-forward (E-C\&F) scheme \cite{nazer2016expanding}, which is versatile to distribute UE's power unequally and can tolerate different noise at the targeted AP for ensuring the minimum loss in the computation rate, offers significant performance improvement over the C\&F scheme. Fig.~\ref{fig:e cf} shows the CDF of the achievable sum-rate obtained via C\&F and E-C\&F schemes versus the number of APs with $K=8$, $L=20$, and ${p=200}$ mW. The E-C\&F scheme outperforms the conventional C\&F scheme in terms of achievable sum-rate noticeably since the E-C\&F framework enables optimal transmit power of UEs which facilitates the exploitation of performance gain.

\begin{figure}[t]
\centering
\includegraphics[scale=0.64]{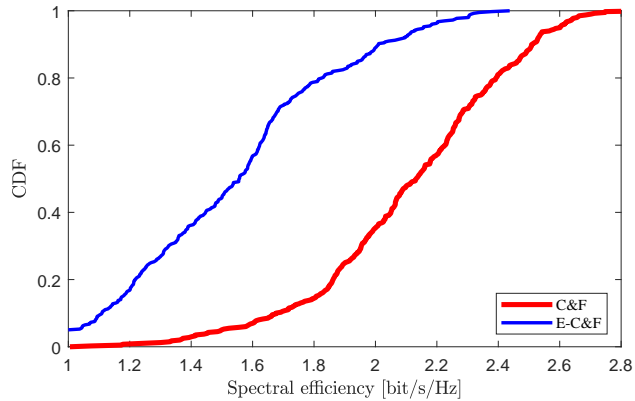}
\caption{CDF of the achievable sum-rate with {C\&F \cite{zhang2019expanded}} and {E-C\&F \cite{nazer2016expanding}} schemes.}
\label{fig:e cf}
\end{figure}

\begin{figure}[t]
\centering
\includegraphics[scale=0.64]{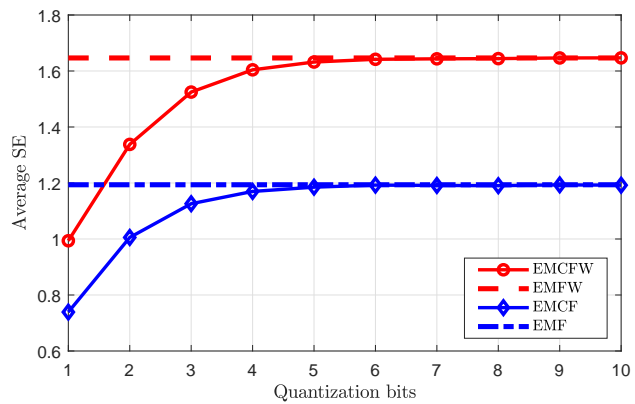}
\caption{Average SE per UE versus the number of quantization bits with {EMCF \cite{bashar2019energy}}, EMF, {EMCFW \cite{femenias2020fronthaul}}, and EMFW schemes.}
\label{fig:quantization}
\end{figure}

Although quantizing is an efficient and straightforward solution to reduce the fronthaul load, the uniform quantization method can induce the quantization error and achieves a considerable performance loss. Besides, it is shown that performances of the different types of transmission are not the same because they require different fronthaul rate allocations for CSI and/or data signals transmitted to CPU, and their AP signal processing capabilities are different.

{To compare the six transmission strategies by highlighting the required processing at APs and the CPU, Table \ref{tab:quantization} is provided.} Besides, in Fig.~\ref{fig:quantization}, we plot the average SE with MR combining at the APs, as a function of the number of the quantization bits for EMCF, EMF, EMCFW, and EMFW, where EMF and EMFW refer to the case of using perfect ADCs in EMCF and EMCFW, respectively, and hence no compression is done at the APs. It can be seen from Fig.~\ref{fig:quantization} that EMCFW provides a larger SE than EMCF. This can be explained that the optimized receiver filter coefficients step in EMCFW maximizes the SNR. Besides, via numerical results in \cite{bashar2019energy,femenias2020fronthaul}, it is shown that the EMCF can outperform the other schemes when they apply UatF bounding (not the CSI-based ones), and ECF strategy also outperforms the CFE. Thus, to decrease the deployment cost of the network, dummy APs which only compress and forward the received signals are preferred at the cost of performance loss. So, there is a tradeoff between performance gain and implementation costs which must be taken into account.

\begin{table*}[t!]
  \centering
  \fontsize{9}{12}\selectfont
  \caption{Comparison of the Quantization Schemes.}
  \label{tab:quantization}
    \begin{tabular}{ !{\vrule width1.2pt}  m{1.8cm}<{\centering} !{\vrule width1.2pt}  m{3.5cm}<{\centering} !{\vrule width1.2pt}  m{3.5cm}<{\centering} !{\vrule width1.2pt} m{3cm}<{\centering} !{\vrule width1.2pt}  m{1.8cm}<{\centering}  !{\vrule width1.2pt} }

    \Xhline{1.2pt}
        \rowcolor{gray!50} \bf Scheme  &  \bf Processing at the APs &  \bf Processing at the CPU & \bf Comb./prec. Design & \bf Compression \cr
    \Xhline{1.2pt}

        CFE \cite{masoumi2019performance,masoumi2019transmission,bashar2020exploiting,bashar2020deep,bashar2018performance,femenias2020fronthaul,bashar2018cell}&  \makecell[c]{CSI compression \\ Data compression} & \makecell[c]{Channel estimation \\ Combining design} & Centralized & Distributed \cr\hline

        ECF \cite{masoumi2019performance,masoumi2019transmission}&  \makecell[c]{Channel estimation \\CSI compression \\ Data compression} & \makecell[c]{ Combining design} & Centralized & Distributed \cr\hline

        EMCF \cite{bashar2019energy}&  \makecell[c]{Channel estimation \\Combining design} & -- & Distributed & Distributed \cr\hline

        EMCFW \cite{bashar2019energy,femenias2020fronthaul,bashar2018cell}&  \makecell[c]{Channel estimation \\Combining design} & \makecell[c]{Receiver filter design} & Distributed & Distributed \cr\hline

        CAP \cite{femenias2019reduced,femenias2019cell,parida2018downlink} &  \makecell[c]{Precoding design \\ Compression} & -- & Centralized & Centralized \cr\hline

        PAC \cite{boroujerdi2019cell} &  \makecell[c]{Compression} & \makecell[c]{Precoding design} & Distributed & Distributed \cr

    \Xhline{1.2pt}
    \end{tabular}
  \vspace{0cm}
\end{table*}

\subsection{Hardware Impairments}\label{subsec:}

\begin{figure}[t]
\centering
\includegraphics[scale=1]{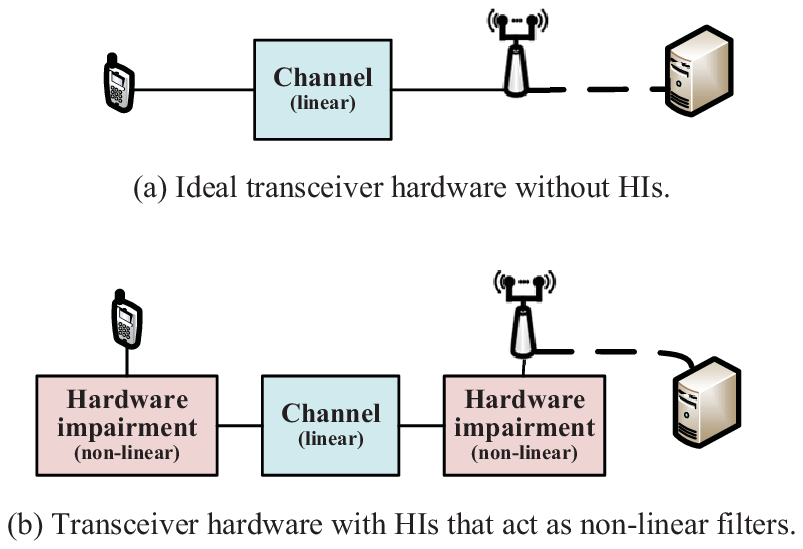}
\caption{A generalized CF mMIMO system with/without HIs.
\label{fig:hardware impairment}}
\end{figure}

To enable a ubiquitous deployment of CF mMIMO, a tradeoff between cost and quality of the transceiver hardware in CF mMIMO should be considered since many antenna elements in APs are deployed, which might increase the deployment cost and energy consumption. A possible countermeasure is to make use of compact low-cost components, which  introduce power amplifier non-linearities, phase noise in local oscillators, amplitude/phase
imbalance in I/Q mixers, and finite-resolution quantization in ADCs.
All these non-idealities are referred to as \emph{hardware impairments} (HIs).
Most existing works on CF mMIMO neglect the impacts of the HIs by modeling the wireless communication channels as linear filters, as shown in Fig.~\ref{fig:hardware impairment}(a).
Although this ideal model can be used to devise analog or digital compensation algorithms that could substantially mitigate the impacts of the HIs, the residual HIs will still exist due to modeling inaccuracies and the destructive nature of some HIs.
Instead, the non-ideal transceiver hardware can be modeled as non-linear memoryless filters to provide better insights for the practical implementation, as shown in Fig.~\ref{fig:hardware impairment}(b).

Research on the impact that residual HIs have on the data rate performance of the CF mMIMO systems has been made in \cite{[69],[237],[102],[195],[225], Bjornson2015b}.
To be specific, authors in \cite{[69]} quantified the performance of both uplink and downlink CF mMIMO with the classical additive hardware distortion model.
By exploiting the MR processing, SE and EE were derived in closed-form.
With these tractable expressions, hardware-quality scaling laws were presented, which proved that the detrimental effect of the HIs at the APs vanishes as the number of APs increases.
Moreover, a max-min power control algorithm was proposed to maximize the minimum UE data rate.
Mitigating the impacts of the HIs by using different levels of AP cooperation was investigated in \cite{[237]}, where the APs could perform data decoding fully distributively or by exploiting LSFD.
The results revealed that the LSFD could provide the largest SE under the HIs.
\cite{[102],[195]} focused on the impact of employing low-resolution ADCs on the CF mMIMO.
A simple asymptotic approximation for the achievable rate was derived by considering the effects of AP and UE number, antenna number per AP, and ADC resolution.
The results showed that the achievable rate of an arbitrary UE converges to a finite limit, independent of the ADC resolution of the APs, as the number of APs goes to infinity.
Additionally, an ADC resolution bits allocation scheme was proposed to maximize the sum rate given a fixed total ADC resolution bits.
Authors in \cite{[225]} considered the physical layer security in a CF mMIMO system with HIs, where a lower bound for the ergodic secrecy rate in the presence of pilot spoofing attack and imperfect CSI is derived.
Moreover, an optimal power allocation scheme was obtained to maximize the achievable secrecy rate using the continuous approximation and path-following algorithms.
Analytical results revealed that the hardware-quality scaling law is almost inapplicable for secure transmission in CF mMIMO system except for some particular scaling factors.
An article that predates the CF mMIMO area but applies to the same scenario is \cite{Bjornson2015b}, which assessed the impact of HIs on scalable CF mMIMO systems by also considering the effect of phase noise.

\section{Future Research Directions}\label{sec:future}

In this section, we briefly highlight some major open problems and research challenges to be addressed in future work.

\subsection{Multiple CPUs and Practical Fronthaul Topology}

Much of the algorithmic design for CF mMIMO has been developed to be transparent to the topology of the underlying network architecture \cite{cellfreebook}, to make it applicable when having one or multiple CPUs, and having parallel or sequential fronthaul connections. In the canonical case, the network comprises many distributed APs with independent cables to the single CPU, also known as a star topology. Still, it is unlikely that geographically large networks will be deployed in that manner.
There might be multiple CPUs that are connected to disjoint subsets of the APs \cite{[104]}, thus if two APs that belong to different CPUs are cooperating, the fronthaul signaling will have to involve multiple APs.
Moreover, the radio-stripes topology has been proposed in \cite{[103]}, where a set of APs are deployed along a fronthaul cable and, thus, have a sequential connection to the CPU. This design is motivated by the practical need for limiting the total cable lengths and opens up research questions related to how the communication algorithms can be adapted to exploit the finer details of the fronthaul topology. For example, some centralized processing schemes can be implemented sequentially \cite{Shaik2021a}.
There are many open research challenges related to distributing the signal processing over multiple CPUs and adapting the algorithms to exploit the structure.

Each fronthaul connection will have a limited capacity. The leading theory for CF mMIMO has been developed under the assumption of infinite fronthaul connections, but with a general awareness that one needs to limit the number of signals that are transmitted between APs and CPUs to achieve scalability \cite{cellfreebook}.
When dealing with a practical capacity-limited fronthaul, one must consider the tradeoff between precision and the number of conveying signals.
As discussed in Section~\ref{subsec:fronthaul}, it then matters where the processing is done: Signals that are measured at the APs can either be processed there at full precision or elsewhere with reduced precision. When it comes to signal compression for fronthaul signaling, one can either consider model-aided or data-driven methods, where the latter can make use of autoencoder methodologies.


\subsection{Synchronization}

Coherent signal processing is possible only if the APs maintain a sufficiently accurate relative timing and phase synchronization.
The network might have an absolute time (phase) reference, but the APs are unsynchronized.
Therefore, AP synchronization and TDD reciprocity calibration are two critical problems to enable CF mMIMO.
Suppose each AP has a local oscillator, and the wired fronthaul network cannot provide a sufficiently accurate common time and frequency reference. In that case, such synchronization must occur via OTA signaling \cite{[103]}. AirSync, which provides timing and phase synchronization accuracy, has been implemented in distributed mMIMO \cite{[sync1]}. Specifically, it detects the slot boundary such that all APs are time-synchronous within a cyclic prefix (CP) of the OFDM modulation and predicts the instantaneous carrier phase correction along with the transmit slot such that all transmitters maintain their coherence.
To limit the reciprocity and synchronization errors, a synchronization process needs to be applied at regular intervals. High-precision inter-node clock synchronization is a prerequisite for joint processing of distributed mMIMO \cite{[sync2]}. All radio access unit (RAU) clocks in the system are assigned by the master node through IEEE 1588 PTPv2. In practice, a global position system (GPS) can also be used for more precise synchronization.
More generally, the communication theory that underpins CF mMIMO assumes a perfect timing synchronization, which is physically impossible over a large network, even if the clocks are synchronized. Hence, there is room for theoretical advancements as well.

\subsection{Mobile Edge Computing}

{The previous subsections focused on} where they carry out the lower-layer processing in a communication network.
A related concept is mobile edge computing (MEC), where the computation/storage resources of the higher layers in the network are pushed to the edge to alleviate the burden of core networks \cite{xu2019modeling,chen2021wireless}.
MEC will naturally reduce the latency since processing is moved closer to the UEs and use general-purpose cloud computing hardware that can be co-located with CPUs.
The user-edge-cloud architecture conceived for MEC perfectly matches the UE-AP-CPU architecture of CF mMIMO, making MEC and CF mMIMO a perfect fit \cite{[197],ke2020massive}.
Edge nodes (APs and CPUs) equipped with computation/storage capability could deal with UEs' computation
and content requests, and consequently, reduce the transmission delay and requirement of the fronthaul/backhaul connection capacities. The research into this direction is in its infancy.

\subsection{Enabling Federated Learning}

Apart from using downlink to improve the channel estimation or resource allocation in the CF mMIMO system, as described earlier in this survey, a wireless network can also be part of the infrastructure used when implementing machine learning algorithms for other applications.
The federated learning (FL) concept can facilitate collaborative learning processes of complex models among the distributed devices, keeping their local training data and control privacy.
The signals sent from the UEs over the wireless network are suggested local model updates (i.e., the model's weights), aggregated at the core of the network, where global model updates are determined and broadcasted to the UEs.
The distributed processing manner of FL naturally fits CF mMIMO, which makes CF mMIMO an enabler of FL \cite{[218],vu2020user}.

\subsection{Multi-Antenna UEs}

The main theory for CF mMIMO has been developed for single-antenna UEs, even though contemporary UEs have at least two antennas, and future devices will feature even larger arrays when operating in the mmWave bands. First steps towards considering CF mMIMO with multi-antenna UEs are found in \cite{Li2016b,[113],[194]}. In general, the multiple antennas can either be used for spatial multiplexing of multiple streams per UE (up to one per antenna) or for improved precoding/combining that mitigates interference \cite{Bjornson2013b}. While the achievable SE can be quantified using existing methods, there are many open resource allocation questions related to power allocation, pilot assignment, and precoding/combining design.

\subsection{Channel Estimation and Prediction Beyond the Block Fading Model}
In practice, the wireless channels vary continuously over time and frequency, not in the block-fading manner described in this survey and most of the theoretical works on this topic. On the one hand, the coherence block size can always be dimensioned in a conservative manner such that the channels are indeed approximately constant within each block. On the other hand, underlying physical rules dictate how the channel can evolve and frequently. By exploiting such properties, using model-aided or data-driven approaches \cite{jiang2020deep}, the communication performance can be significantly improved: pilots can be transmitted less frequently, and/or the CSI quality can be increased. Prior work on this topic has been done in the cellular mMIMO field \cite{kim2020massive}, which potentially can be adapted to cover CF mMIMO.

\subsection{Integrated Sensing and Computing}
In future wireless communications, integrated sensing and communication (ISAC) will be a paradigm change. Some promising ISAC-like dual-functional radar-communication (DFRC) system has attracted substantial attention, where joint radar sensing and multi-user communication can be simultaneously implemented \cite{liu2018mu,liu2020joint}. And it is also interesting to investigate a DFRC-based system with CF mMIMO for the feature of uniform coverage. However, signal processing and system design will be the key challenges. Besides, the great demand for computation is also an issue to address. Therefore, the DFRC system based on scalable CF mMIMO will be investigated in the future.

\section{Conclusion and Lessons learned}\label{sec:conclusion}

Exploiting densification and decentralization to boost the user-experienced data rates and realize a ubiquitous service is an irresistible general trend for future wireless communications.
CF mMIMO represents an attempt to reach this promising prospect by coordinating dense serving antennas in a decentralized CF approach, which greatly squeezes the potential of the multiple antenna technology so that it dynamically achieves the best performance with the available resources.
In particular, the dense deployment of the serving antennas in CF mMIMO will result in strong macro-diversity from a significantly smaller average distance between a UE and its closest antennas, while the joint signal processing and scheduling among the distributed antennas achieve array gains and spatial interference suppression which substantially reduce the QoS variations within the coverage area.

In this paper, we have presented a comprehensive review of the concepts and techniques proposed for CF mMIMO systems.
First, we gave the motivation for CF mMIMO and provide a brief introduction of CF mMIMO itself and the other technologies related to it.
Then we used a section to briefly provide a tutorial about the technical foundations of CF mMIMO including the transmission procedure and mathematical system model.
The core of the paper provides an extensive survey on the state-of-the-art schemes and algorithms available in the literature for the resource allocation and signal processing (i.e., channel estimation, combining and precoding, user access and association, and power control) and practical implementation (i.e., fronthaul limitation and hardware impairment) in CF mMIMO systems.
We then highlighted the open research areas for CF mMIMO (e.g., multiple CPUs cooperation, MEC, enabling machining learning, etc.) and proposed potential approaches for solutions.

Although the research on this topic is still in the exploratory phase, the primary demonstrations, field tests, and prototypes of CF mMIMO systems have been ongoing across different projects in academia and industry.
How to realize the scalable intelligent system deployment with new mathematical tools, new applications, and new standardizations becomes a very attractive open issue for all researchers in this field.
Though many challenges remain to address, CF mMIMO shows great potential to meet the ubiquitous high QoS demands of the 6G communications.
In the foreseeable future, research on CF mMIMO will continue to mature. With no doubt, this technology with its concepts will open up new frontiers in wireless services and applications.

%

\bibliographystyle{IEEEtran}
\bibliography{IEEEabrv,Ref_cell_free_update0928}

\end{document}